%% file: main.tex
\documentclass[twocolumn,10pt]{IEEEtran}

\usepackage[acronym]{glossaries}
\loadglsentries{glossary}
\usepackage{amsmath,amssymb,amsthm}
\usepackage{graphicx}
\usepackage{comment}
\usepackage{xcolor}
\usepackage{subfigure}
\usepackage{cases}
\usepackage{booktabs}
\usepackage{algorithm}
\usepackage{algcompatible}
\usepackage{multirow}

\newtheorem{proposition}{Proposition}
\newtheorem{remark}{Remark}

\begin{document}

\title{Analog Computing for Signal Processing\\and Communications -- Part~I:\\Computing with Microwave Networks}

\author{Matteo~Nerini,~\IEEEmembership{Member,~IEEE},
        Bruno~Clerckx,~\IEEEmembership{Fellow,~IEEE}

\thanks{Part of this work has been accepted by the 2025 IEEE Global Communications Conference (GLOBECOM) \cite{ner25-1}. (Corresponding author: Bruno Clerckx.)}
\thanks{This work has been partially supported by UKRI grant EP/Y004086/1, EP/X040569/1, EP/Y037197/1, EP/X04047X/1, and EP/Y037243/1.}
\thanks{Matteo Nerini and Bruno Clerckx are with the Department of Electrical and Electronic Engineering, Imperial College London, London SW7 2AZ, U.K. (e-mail: \{m.nerini20, b.clerckx\}@imperial.ac.uk).}
\thanks{Bruno Clerckx is also with Kyung Hee University, Seoul, Korea.}}

\maketitle

\begin{abstract}
Analog computing has been recently revived due to its potential for energy-efficient and highly parallel computations.
In this two-part paper, we explore analog computers that linearly process microwave signals, named \glspl{milac}, and their applications in signal processing and communications.
In Part~I of this paper, we model a \gls{milac} as a multiport microwave network with tunable impedance components, enabling the execution of mathematical operations by reconfiguring the microwave network and applying input signals at its ports.
We demonstrate that a \gls{milac} can efficiently compute the \gls{lmmse} estimator and matrix inversion, with remarkably low computational complexity.
Specifically, a matrix can be inverted with complexity growing with the square of its size.
We also show how a \gls{milac} can be used jointly with digital operations to implement sophisticated algorithms such as the Kalman filter.
To enhance practicability, we propose a design of \gls{milac} based on lossless impedance components, reducing power consumption and eliminating the need for costly active components.
In Part~II of this paper, we investigate the applications of \glspl{milac} in wireless communications, highlighting their potential to enable future wireless systems by executing computations and beamforming in the analog domain.
\end{abstract}

\glsresetall


\section{Introduction}

Computing devices play a pivotal role in our society, serving as the backbone of the technological advancement.
As our world becomes increasingly technology-driven and interconnected, we continuously use computing devices to process data, automate tasks, solve problems, and exchange information.
Modern computers involve digital electronic circuits operating on discrete values (binary) to encode and manipulate data.
Digital circuits and signals are widespread today because of their numerous benefits over their analog counterparts \cite{pro96}, such as versatility, more precision and noise robustness, flexibility in storing and processing data, and security.

Despite their popularity, digital computers face two key limitations that may hinder their ability to meet the ever-increasing demand for computational power.
First, they are characterized by high power consumption, especially due to the presence of power-hungry \glspl{adc} and \glspl{dac}.
Second, their speed is limited by the clock of the digital processors and by the rate of \glspl{adc} and \glspl{dac}.
To overcome these two limitations, analog computing has been recently revived as it enables energy-efficient and massively parallelized computations \cite{cal13}.
As defined in \cite{cal13}, analog computing refers to manipulating analog signals in real time to perform specific operations.
Modern analog computing devices have been proposed in recent years based on optical systems \cite{amb10,sol15}, metamaterials and metasurfaces \cite{sil14,zan21}, and resistive memory arrays \cite{iel18,wan20}.

Analog computers based on optical systems, or optical analog computers, leverage the intrinsic parallelism of light propagation to process data, represented as optical signals, at an exceptionally high rate \cite{amb10,sol15}.
Optical analog computers have been designed to efficiently perform linear operations \cite{mat19} as well as compute the inverse of a given matrix \cite{wu14}.
More complicated tasks have also been performed through optical analog computers, specifically related to image processing and machine learning inference, as reviewed in \cite{wet20}.
In \cite{zho20}, a 2-dimensional spatial differentiator has been built to transform the input optical signal into its second-order derivative, enabling edge detection within an image.
Machine learning inference in the analog domain has been demonstrated in \cite{lin18}, where an optical system based on multiple diffractive layers was proposed to emulate the forward propagation of a \gls{nn}.
However, this system is purely linear, and as pointed out in \cite{wei18}, it lacks the nonlinear activations that are pivotal in \gls{nn}.
Diffractive optics has also been used to build a hybrid optical-electronic convolutional \gls{nn}, in which the convolution operation is performed in the analog domain through a phase mask \cite{cha18}.
Although optical analog computers have demonstrated significant advantages in low-power and ultra-fast processing, they typically rely on large-scale architectures, as their size cannot be reduced below the operational wavelength.

To realize miniaturized analog computers at scales smaller than the operational wavelength, metamaterials have emerged as artificial materials designed to have unusual \gls{em} properties \cite{sil14,zan21}.
Because of their artificially induced properties, metamaterials allow us to manipulate \gls{em} signals within sub-wavelength scales.
Analog computers based on metamaterials have demonstrated the possibility of solving mathematical equations \cite{moh19}, including multiple equations in parallel exploiting input signals at different frequencies \cite{cam21}.
Particular interest has been dedicated to the 2-dimensional counterpart of metamaterials, namely metasurfaces, further reducing the dimensions of the computing devices.
Metasurfaces can be engineered to apply specific transfer functions to the spatial and frequency components of the incident wave, including differentiation, integration, and convolution \cite{sil14,zan21,sol22}.
To perform these mathematical operations, various metasurfaces technologies have been developed, such as space-time digital metasurfaces \cite{raj21} and nonlocal metasurfaces having an angular-dependent response \cite{kwo18,sha23}.
Beyond linear transformations \cite{del18}, metasurfaces have been designed to perform edge detection and pattern recognition, which are crucial operations in image processing \cite{zhu17,wan22}.
The capability of metasurfaces to perform matrix-vector multiplications in the analog domain has been exploited to implement programmable \glspl{nn} \cite{mom23}, performing recognition tasks \cite{wu21-1} and reinforcement learning tasks \cite{liu22}.
Metasurfaces are also expected to play a crucial role in future wireless communications due to their ability to perform computations directly in the wave domain \cite{li22}.
Possible applications of computing metasurfaces in wireless communications include performing sensing tasks and enhancing physical layer security \cite{yan23, oma25}.

Besides metamaterials and metasurfaces, resistive memory arrays have also emerged as a prominent technology for building modern analog computing devices \cite{iel18,wan20}.
Particularly noteworthy are crosspoint memory arrays, in which resistive memory cells are located at the intersections of horizontal and vertical lines in a matrix shape \cite{zho19}.
The values stored in these memory cells can naturally represent the entries of a matrix, enabling the memory array to perform matrix operations, such as matrix-vector multiplications and solving linear algebra problems, directly in the analog domain \cite{sun22,man23,man23b,pan24}.
This approach, referred to as in-memory computing or analog matrix computing, has demonstrated the capability to execute matrix-vector multiplications with arbitrarily high precision while consuming less energy compared to digital computing \cite{son24}.
The parallelism offered by resistive memory arrays in computing matrix-vector multiplications has been harnessed to effectively implement and accelerate \glspl{nn} \cite{bur17,yu18,hae19} and image processing \cite{li18}.
Furthermore, analog matrix computing via resistive memory arrays has been used to advance future \gls{mimo} communications \cite{man22,wan23}.

Several communities have developed modern analog computing solutions to accelerate signal processing operations \cite{amb10}-\cite{liu22}, \cite{sun22}-\cite{hae19} as well as to foster wireless communications \cite{li22,yan23,oma25,man22,wan23}.
However, the efforts of these communities are often siloed, with limited interdisciplinarity, and this fragmentation has led to a lack of a unified model and shared theoretical understanding.
In this two-part paper, we aim to fill this gap, by focusing on analog computers that linearly process microwave signals.
We develop the model of these analog computers by using rigorous microwave theory, characterize their input-output relationship, and explore their application to signal processing (in Part~I) and wireless communications (in Part~II).
Specifically, the contributions of Part~I of this paper are as follows.

\textit{First}, we investigate the fundamental computational limits of a reconfigurable microwave network made of linear components.
To this end, we derive a general model for an analog computer that processes microwave signals satisfying the superposition principle, which we denote as \gls{milac}.
A \gls{milac} can be regarded as a microwave network having multiple ports, where input signals are applied and output signals are read.
Once the input signals are applied, they propagate at light speed within the microwave network, generating the output signals.
In this way, specific operations can be instantly computed by the \gls{milac}, requiring no analog-digital conversion or digital computations.
This is the first effort in providing a theoretical framework unifying previous works \cite{amb10}-\cite{wan23}.

\textit{Second}, we characterize the output of a \gls{milac} as a function of both the input signals and the impedance components constituting its microwave network.
Specifically, we derive the expression of the output signals by using microwave theory, and regarding the \gls{milac} as a multiport microwave network.
The obtained expressions involve matrix-matrix products and matrix inversions, which are computationally expensive operations in the digital domain, as their complexity scales cubically with the matrix size.
Nevertheless, they can be instantly computed by a \gls{milac} once the microwave network is properly reconfigured and the input signals are applied.
While it is known that specific matrix operations can be computed in the analog domain \cite{amb10}-\cite{wan23}, we provide the general expression for the functions computable with microwave signals by a linear network.

\textit{Third}, we show that a \gls{milac} can be used as a computing device to efficiently calculate popular mathematical operations, namely the \gls{lmmse} estimator and the inverse of a given matrix.
Compared to a digital computer, a \gls{milac} can calculate these two operations with a reduction in computational complexity by a factor of $2.5\times 10^4$ and $5.5\times 10^3$, respectively, when the input size is $8192$.
We define a set of primitive operations computable with a \gls{milac}, and show how they can be sequentially used jointly with digital operations to efficiently implement more sophisticated algorithms, such as the Kalman filter.
Specifically, a Kalman filter can be implemented with a \gls{milac} by reducing the computational complexity by $1.1\times 10^4$ times, with input size being $8192$.
The sequential use of primitive operations computed in the analog domain is less explored in literature \cite{amb10}-\cite{wan23}, and can further extend the set of possible functions implementable with analog computers.

\textit{Fourth}, we propose a practical design to implement a \gls{milac}, relying on a microwave network solely constituted of lossless tunable impedance components, i.e., with purely imaginary values instead of complex.
The use of such components is expected to reduce losses and avoid costly active \gls{rf} components, providing the initial guidelines for the practical implementation of a \gls{milac}.
We show that such a \gls{milac} made of imaginary impedance components can execute the same operations as a \gls{milac} built with complex impedance components.

\textit{Organization}:
In Section~\ref{sec:model}, we present the general model a \gls{milac}.
In Section~\ref{sec:analysis}, we analyze a \gls{milac} and characterize its output as a function of its input.
In Section~\ref{sec:lmmse-inv}, we show how a \gls{milac} can efficiently compute the \gls{lmmse} estimator and invert a matrix.
In Section~\ref{sec:kalman}, we show that a \gls{milac} can also efficiently compute the Kalman filter.
In Section~\ref{sec:implementation}, we propose a practical implementation of a \gls{milac}, based on lossless and reciprocal tunable impedance components.
Finally, Section~\ref{sec:conclusion} concludes Part~I of this paper.

\textit{Notation}:
Vectors and matrices are denoted with bold lower and bold upper letters, respectively.
Scalars are represented with letters not in bold font.
$\Re\{a\}$, $\Im\{a\}$, and $\vert a\vert$ refer to the real part, imaginary part, and absolute value of a complex scalar $a$, respectively.
$\mathbf{a}^T$, $\mathbf{a}^H$, $[\mathbf{a}]_{i}$, and $\Vert\mathbf{a}\Vert_2$ refer to the transpose, conjugate transpose, $i$th element, and $l_{2}$-norm of a vector $\mathbf{a}$, respectively.
$\mathbf{A}^T$, $\mathbf{A}^H$, $[\mathbf{A}]_{i,k}$, $[\mathbf{A}]_{:,k}$, and $\Vert\mathbf{A}\Vert_F$ refer to the transpose, conjugate transpose, $(i,k)$th element, $k$th column, and Frobenius norm of a matrix $\mathbf{A}$, respectively.
$\mathbb{R}$ and $\mathbb{C}$ denote the real and complex number sets, respectively.
$j=\sqrt{-1}$ denotes the imaginary unit.
$\mathbf{I}_N$ and $\mathbf{0}_N$ denote the identity matrix and the all-zero matrix with dimensions $N\times N$, respectively, and $\mathbf{0}_{M\times N}$ denote the all-zero matrix with dimensions $M\times N$.

\begin{figure}[t]
\centering
\includegraphics[width=0.42\textwidth]{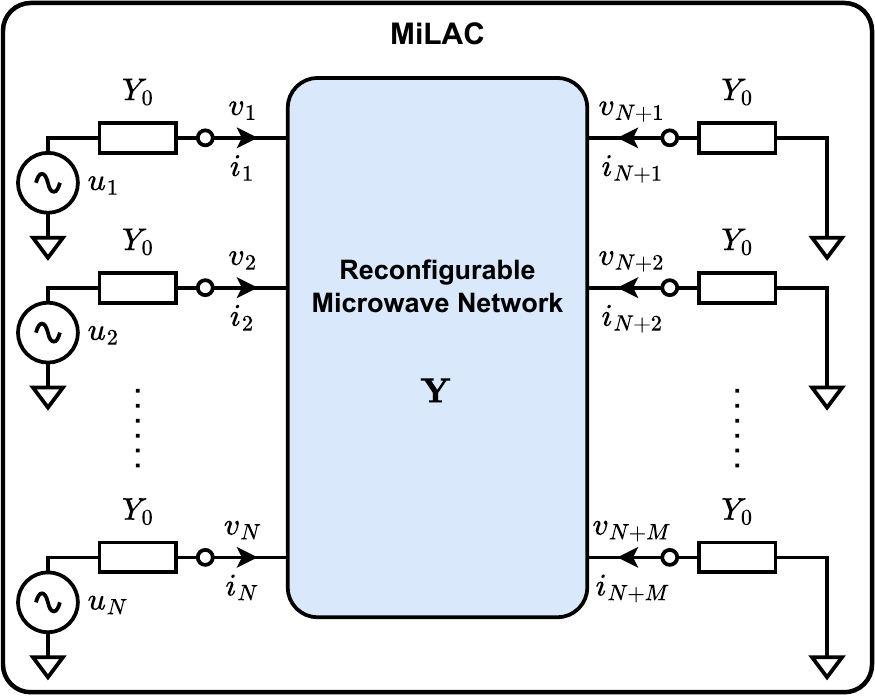}
\caption{Representation of a $P$-port MiLAC with input on $N$ ports, where $P=N+M$.}
\label{fig:ac}
\end{figure}

\section{Model of a Microwave Linear Analog Computer}
\label{sec:model}

In this section, we introduce and model a class of analog computers, that we denote as \gls{milac}.
We define a microwave analog computer as an analog computer that processes microwave \gls{em} signals, i.e., with frequencies from 300~MHz to 300~GHz \cite{poz11}.
In addition, we introduce a linear analog computer as an analog computer that satisfies the superposition principle, i.e., the output caused by multiple inputs is the sum of the outputs caused by each input individually.
Thus, a \gls{milac} is defined as an analog computer that is both microwave and linear.
In the following discussion, it is important to specify the microwave regime as all results are derived by applying microwave theory \cite{poz11}.

A \gls{milac} can be modeled in general as a linear \textit{microwave network} having multiple terminals, or ports, where \textit{input signals} are applied and \textit{output signals} are read.
To enable the \gls{milac} to compute different functions, its microwave network must be reconfigurable, which can be achieved through the use of tunable impedance components, or, equivalently, \textit{tunable admittance components}.
In the following, we model a \gls{milac}, detailing its four building blocks: microwave network, input signals, output signals, and tunable admittance components.

\subsection{Microwave Network}

The circuit implementing a \gls{milac} is in general a linear microwave network, following the definition of \gls{milac}.
According to microwave network theory \cite[Chapter 4]{poz11}, a linear microwave network having a certain number of ports $P$ can be characterized by its impedance matrix $\mathbf{Z}\in\mathbb{C}^{P\times P}$, which linearly relates the voltages and currents at the $P$ ports through
\begin{equation}
\mathbf{v}=\mathbf{Z}\mathbf{i},\label{eq:Z}
\end{equation}
where $\mathbf{v}=[v_1,\ldots,v_P]^T\in\mathbb{C}^{P\times1}$ and $\mathbf{i}=[i_1,\ldots,i_P]^T\in\mathbb{C}^{P\times1}$ are the voltage and current vectors at the ports, with $v_p\in\mathbb{C}$ and $i_p\in\mathbb{C}$ being the voltage and current at the $p$th port, for $p=1,\ldots,P$.
\footnote{Throughout the paper, voltage and current signals are narrowband basspand signals represented through their complex-valued baseband equivalent.
For example, for the voltage at the $p$th port, $v_p\in\mathbb{C}$ is the baseband equivalent, while the actual passband signal is $\Re\{v_pe^{j2\pi f_ct}\}$, where $f_c$ is the carrier frequency.
In other words, the complex number $v_p=Ae^{j\varphi}$ represents the signal $A\cos(2\pi f_ct+\varphi)$, or equivalently, the complex number $v_p=I+jQ$ represents the signal $I\cos(2\pi f_ct)-Q\sin(2\pi f_ct)$.}
Alternatively, a linear microwave network with $P$ ports can also be represented by its admittance matrix $\mathbf{Y}\in\mathbb{C}^{P\times P}$, given by $\mathbf{Y}=\mathbf{Z}^{-1}$, relating the currents and voltages at the $P$ ports as
\begin{equation}
\mathbf{i}=\mathbf{Y}\mathbf{v},\label{eq:Y}
\end{equation}
equivalently to \eqref{eq:Z} \cite[Chapter 4]{poz11}.
Impedance and admittance matrices are equivalent representations of a linear microwave network, allowing its analysis regardless of its internal circuit topology.

\subsection{Input Signals}

Assume that the \gls{milac} receives the input on the first $N$ ports of its reconfigurable microwave network, with $N\leq P$.
This comes with no loss of generality, as the ports can always be reordered to ensure the input is applied on the first $N$ ports.
In addition, if an input needs to be applied on an internal node of the microwave network, an additional port can be opened, making this model general.
The input is applied by $N$ voltage sources with their series impedance $Z_0$, e.g., $Z_0=50~\Omega$, or, equivalently, series admittance $Y_0=Z_0^{-1}$, as shown in Fig.~\ref{fig:ac}.
Denoting as $u_n\in\mathbb{C}$ the voltage of the $n$th voltage source, for $n=1,\ldots,N$, the input vector is given by $\mathbf{u}=[u_1,\ldots,u_N]^T\in\mathbb{C}^{N\times1}$.
Following Ohm's law, the input voltage $u_n$ is related to the voltage $v_n$ and current $i_n$ at the $n$th port by $i_n=Y_0(u_n-v_n)$, for $n=1,\ldots,N$, giving
\begin{equation}
\mathbf{i}_1=Y_0\left(\mathbf{u}-\mathbf{v}_1\right),\label{eq:input}
\end{equation}
where we introduced $\mathbf{v}_1=[v_1,\ldots,v_N]^T\in\mathbb{C}^{N\times1}$ and $\mathbf{i}_1=[i_1,\ldots,i_N]^T\in\mathbb{C}^{N\times1}$ as the vectors of the voltages and currents at the first $N$ ports of the network, respectively.

\subsection{Output Signals}

The output of the \gls{milac} is read on all the $P$ ports of its reconfigurable microwave network, i.e., is given by the vector $\mathbf{v}=[v_1,\ldots,v_P]^T$.
If an output needs to be read on an internal node of the microwave network, an additional port can be added, confirming the generality of our model.
Since the number of input ports $N$ is never higher than the number of total ports $P$, we can write $P=N+M$, where $M$ is the number of ports with no input signal.
To read the output voltage on these $M$ ports without causing undesired reflections within the microwave network, these $M$ ports must be perfectly matched to $Z_0$, i.e., they are terminated to ground via an impedance $Z_0$, or, equivalently, an admittance $Y_0=Z_0^{-1}$, as shown in Fig.~\ref{fig:ac}.
Thus, according to Ohm's law, the voltage and current at the $(N+m)$th port are related by $i_{N+m}=-Y_0v_{N+m}$, for $m=1,\ldots,M$, where the negative sign arises from the direction of the current $i_{N+m}$.
This yields
\begin{equation}
\mathbf{i}_2=-Y_0\mathbf{v}_2,\label{eq:output}
\end{equation}
where we introduced $\mathbf{v}_2=[v_{N+1},\ldots,v_{N+M}]^T\in\mathbb{C}^{M\times1}$ and $\mathbf{i}_2=[i_{N+1},\ldots,i_{N+M}]^T\in\mathbb{C}^{M\times1}$ as the vectors of voltages and currents at the last $M$ ports of the network, respectively.

\subsection{Tunable Admittance Components}
\label{sec:tunable}

The microwave network of the \gls{milac} must be implemented with tunable admittance components, such that its impedance matrix $\mathbf{Z}$ and admittance matrix $\mathbf{Y}$ are reconfigurable, allowing the computation of different operations in the analog domain.
As impedance and admittance matrices are equivalent representations of a microwave network, we use the admittance matrix $\mathbf{Y}$ representation in this study since it can be directly related to the tunable components of the microwave network, as will become clear in the remainder of this section.

\begin{figure}[t]
\centering
\includegraphics[width=0.42\textwidth]{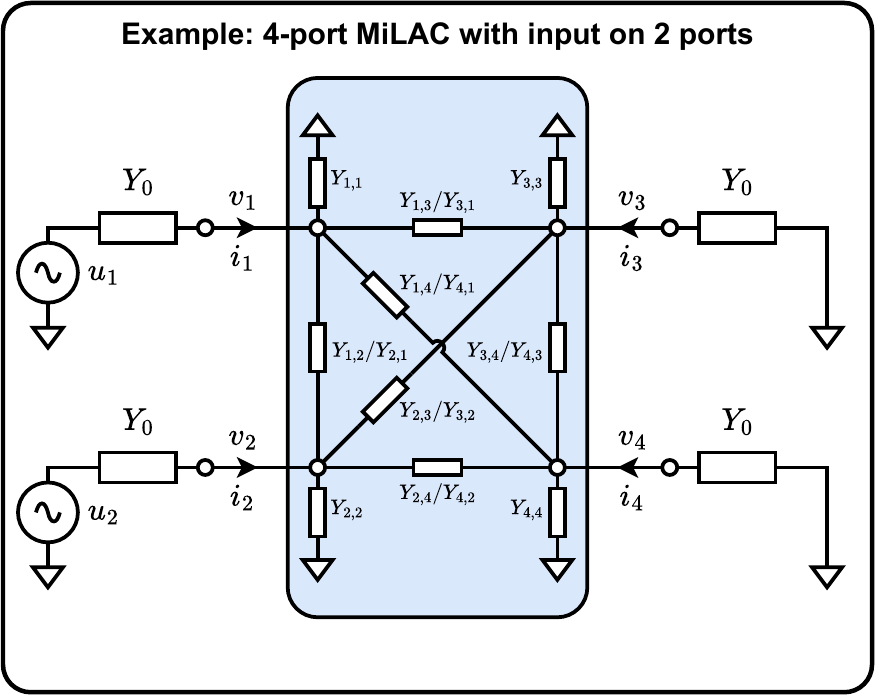}
\caption{Representation of a $4$-port MiLAC with input on $2$ ports.}
\label{fig:ac-4-port}
\end{figure}

To realize a \gls{milac} with maximum flexibility, i.e., whose admittance matrix $\mathbf{Y}$ can be reconfigured to any complex matrix with dimension $P\times P$, we can implement its microwave network with $P^2$ tunable admittance components connecting each port to ground and to all other ports.
We characterize these components through the lumped-element model.
Specifically, port $k$ is connected to ground through an admittance $Y_{k,k}\in\mathbb{C}$, for $k=1,\ldots,P$, and port $k$ is connected to port $i$ through an admittance $Y_{i,k}\in\mathbb{C}$, $\forall i\neq k$.
An example of such a \gls{milac} with maximum flexibility is reported in Fig.~\ref{fig:ac-4-port}, where we have $P=4$ ports and $N=2$ input signals.

\begin{figure}[t]
\centering
\includegraphics[width=0.48\textwidth]{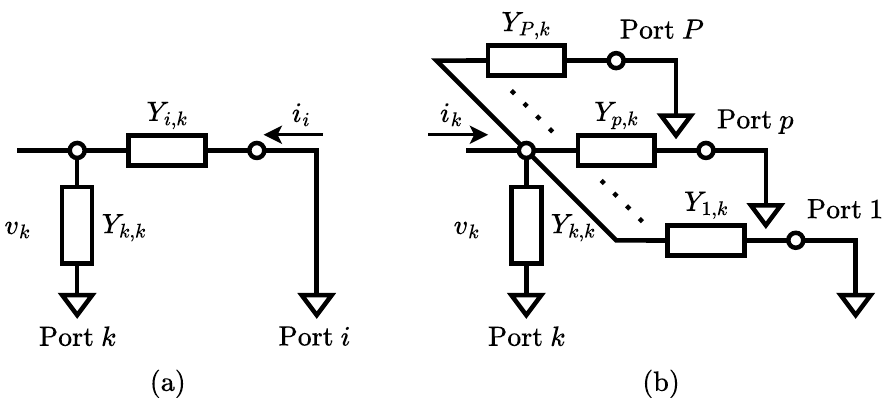}
\caption{Circuit to be analyzed to compute (a) the off-diagonal entry $[\mathbf{Y}]_{i,k}$ and (b) the diagonal entry $[\mathbf{Y}]_{k,k}$ as a function of the tunable admittance components.}
\label{fig:network}
\end{figure}

This topology of tunable admittance components allows us to arbitrarily reconfigure the admittance matrix $\mathbf{Y}$, as it can be shown by following the definition of $\mathbf{Y}$ in \eqref{eq:Y}.
According to \eqref{eq:Y}, $[\mathbf{Y}]_{i,k}$ is given by applying a voltage $v_{k}$ to port $k$, short-circuiting all the other ports, i.e., $v_p=0$, $\forall p\neq k$, measuring the current $i_i$ entering at port $i$, and computing the ratio
\begin{equation}
\left[\mathbf{Y}\right]_{i,k}=\left.\frac{i_i}{v_{k}}\right\vert_{v_p=0,\forall p\neq k},\label{eq:Y-def}
\end{equation}
for $i,k=1,\ldots,P$ \cite[Chapter 4]{poz11}.
Following \eqref{eq:Y-def}, the entries $\left[\mathbf{Y}\right]_{i,k}$ can be computed as a function of the tunable admittance components $\{Y_{i,k}\}_{i,k=1}^{P}$ by separately studying the two cases of off-diagonal entries $\left[\mathbf{Y}\right]_{i,k}$ with $i\neq k$ and diagonal entries $\left[\mathbf{Y}\right]_{i,k}$ with $i=k$.
First, when $i\neq k$, the ratio in \eqref{eq:Y-def} can be determined by studying the circuit in Fig.~\ref{fig:network}(a), since all the other admittance components in the microwave network do not impact the current $i_i$.
Thus, we obtain the off-diagonal entry $\left[\mathbf{Y}\right]_{i,k}$ as $\left[\mathbf{Y}\right]_{i,k}=-Y_{i,k}$ following Ohm's law, where the negative sign arises from the direction of the current $i_i$.
Second, when $i=k$, the ratio in \eqref{eq:Y-def} can be determined by studying the circuit in Fig.~\ref{fig:network}(b), since all the other admittance components do not influence $i_k$.
In this case, the diagonal entry $\left[\mathbf{Y}\right]_{k,k}$ is given by the parallel of the $P$ admittance components $Y_{1,k},\ldots,Y_{P,k}$, i.e., $\left[\mathbf{Y}\right]_{k,k}=\sum_{p=1}^PY_{p,k}$.
By summarizing the two examined cases, we can write
\begin{equation}
\left[\mathbf{Y}\right]_{i,k}=
\begin{cases}
-Y_{i,k} & i\neq k\\
\sum_{p=1}^PY_{p,k} & i=k
\end{cases},\label{eq:Yik-entry}
\end{equation}
for $i,k=1,\ldots,P$, giving the expression of each entry of the admittance matrix $\mathbf{Y}$ as a function of the tunable admittance components $\{Y_{i,k}\}_{i,k=1}^{P}$.
By inverting \eqref{eq:Yik-entry}, we can derive the expression of the tunable admittance components as a function of any given admittance matrix $\mathbf{Y}$ as
\begin{equation}
Y_{i,k}=
\begin{cases}
-\left[\mathbf{Y}\right]_{i,k} & i\neq k\\
\sum_{p=1}^P\left[\mathbf{Y}\right]_{p,k} & i=k
\end{cases},\label{eq:Yik-component-Y}
\end{equation}
for $i,k=1,\ldots,P$, showing that the $P^2$ tunable admittance components $\{Y_{i,k}\}_{i,k=1}^{P}$ can be adjusted to allow $\mathbf{Y}$ to assume any arbitrary value\footnote{Equation \eqref{eq:Yik-component-Y} clarifies the direct relationship between the entries of the admittance matrix $\mathbf{Y}$ and the tunable components of the microwave network.
Similar relationships in closed-form are not available to relate the entries of the impedance matrix $\mathbf{Z}$ and the tunable components.}.

It is also possible to implement a \gls{milac} with a reduced number of tunable admittance components.
However, such a \gls{milac} will experience reduced flexibility since the microwave network cannot be reconfigured to take any arbitrary admittance matrix $\mathbf{Y}$, as a consequence of the reduced circuit complexity.
In this work, we consider a \gls{milac} having maximum flexibility and leave the analysis of the trade-off between circuit complexity and flexibility as future research.

\section{Analysis of a Microwave Linear Analog Computer}
\label{sec:analysis}

We have introduced and modeled the concept of \gls{milac}, highlighting that it consists of a reconfigurable microwave network with multiple ports where input signals are applied and output signals are read.
In this section, we characterize the output $\mathbf{v}$ as a function of the input $\mathbf{u}$ and the admittance matrix of the microwave network $\mathbf{Y}$.
\footnote{Fully-connected microwave networks have also been studied to implement \gls{bd-ris} \cite{she20}.
However, the analysis of this section differs from previous literature since a \gls{milac} receives (and returns) input (and output) signals directly at its ports, while the signals at the ports of a \gls{bd-ris} are typically not considered in a wireless system.}
For better clarity, we distinguish between the cases where the input is applied only to some ports ($N<P$) and to all ports ($N=P$).

\subsection{Analysis of a MiLAC with Input on Some Ports}
\label{sec:analysis1}

In the case the input is applied only to some ports of the reconfigurable microwave network, i.e., $N<P$ and $M>0$, the \gls{milac} is described by the system of linear equations \eqref{eq:Y}, \eqref{eq:input}, and \eqref{eq:output}, given by
\begin{equation}
\begin{cases}
\mathbf{i}&=\mathbf{Y}\mathbf{v}\\
\mathbf{i}_1&=Y_0\left(\mathbf{u}-\mathbf{v}_1\right)\\
\mathbf{i}_2&=-Y_0\mathbf{v}_2
\end{cases},\label{eq:system1}
\end{equation}
which characterizes the \gls{milac} in Fig.~\ref{fig:ac}.
To solve \eqref{eq:system1}, we rewrite it in a more compact form as
\begin{equation}
\begin{cases}
\mathbf{i}&=\mathbf{Y}\mathbf{v}\\
\mathbf{i}&=Y_0\left(\tilde{\mathbf{u}}-\mathbf{v}\right)
\end{cases},\label{eq:system2}
\end{equation}
where we exploited the fact that $\mathbf{v}=[\mathbf{v}_1^T,\mathbf{v}_2^T]^T$ and $\mathbf{i}=[\mathbf{i}_1^T,\mathbf{i}_2^T]^T$, and introduced the vector $\tilde{\mathbf{u}}\in\mathbb{C}^{P\times1}$ as $\tilde{\mathbf{u}}=[\mathbf{u}^T,\mathbf{0}_{M\times1}^T]^T$.
By noticing that \eqref{eq:system2} yields
\begin{equation}
\mathbf{Y}\mathbf{v}=Y_0\left(\tilde{\mathbf{u}}-\mathbf{v}\right),
\end{equation}
we can readily express the output $\mathbf{v}$ as
\begin{equation}
\mathbf{v}=\mathbf{P}^{-1}\tilde{\mathbf{u}},\label{eq:v-tmp}
\end{equation}
where the matrix $\mathbf{P}\in\mathbb{C}^{P\times P}$ is defined as
\begin{equation}
\mathbf{P}=\frac{\mathbf{Y}}{Y_0}+\mathbf{I}_{P},\label{eq:P}
\end{equation}
giving the output $\mathbf{v}$ as a function of the admittance matrix $\mathbf{Y}$ and the input $\mathbf{u}$.
By denoting the inverse of $\mathbf{P}$ as
\begin{equation}
\mathbf{Q}=\mathbf{P}^{-1},
\end{equation}
partitioned as
\begin{equation}
\mathbf{Q}=
\begin{bmatrix}
\mathbf{Q}_{11} & \mathbf{Q}_{12}\\
\mathbf{Q}_{21} & \mathbf{Q}_{22}
\end{bmatrix},
\end{equation}
with $\mathbf{Q}_{11}\in\mathbb{C}^{N\times N}$, $\mathbf{Q}_{12}\in\mathbb{C}^{N\times M}$, $\mathbf{Q}_{21}\in\mathbb{C}^{M\times N}$, $\mathbf{Q}_{22}\in\mathbb{C}^{M\times M}$, the output vector $\mathbf{v}=[\mathbf{v}_1^T,\mathbf{v}_2^T]^T$ given by \eqref{eq:v-tmp} can be further expressed as
\begin{equation}
\begin{bmatrix}
\mathbf{v}_1\\
\mathbf{v}_2
\end{bmatrix}
=
\begin{bmatrix}
\mathbf{Q}_{11}\\
\mathbf{Q}_{21}
\end{bmatrix}
\mathbf{u},\label{eq:v}
\end{equation}
solely depending on the blocks $\mathbf{Q}_{11}$ and $\mathbf{Q}_{21}$ and the input $\mathbf{u}$.
To explicitly express $\mathbf{Q}_{11}$ and $\mathbf{Q}_{21}$ in \eqref{eq:v} as a function of the admittance matrix $\mathbf{Y}$, we exploit the following proposition.

\begin{proposition}
Let $\mathbf{P}\in\mathbb{C}^{P\times P}$ be an invertible matrix, and $\mathbf{P}^{-1}\in\mathbb{C}^{P\times P}$ its inverse, partitioned as
\begin{equation}
\mathbf{P}=
\begin{bmatrix}
\mathbf{A} & \mathbf{B}\\
\mathbf{C} & \mathbf{D}
\end{bmatrix},\;
\mathbf{P}^{-1}=
\begin{bmatrix}
\mathbf{A}^\prime & \mathbf{B}^\prime\\
\mathbf{C}^\prime & \mathbf{D}^\prime
\end{bmatrix},
\end{equation}
where $\mathbf{A},\mathbf{A}^\prime\in\mathbb{C}^{N\times N}$, $\mathbf{B},\mathbf{B}^\prime\in\mathbb{C}^{N\times M}$, $\mathbf{C},\mathbf{C}^\prime\in\mathbb{C}^{M\times N}$, $\mathbf{D},\mathbf{D}^\prime\in\mathbb{C}^{M\times M}$, and $P=N+M$.
Then, we have the following results:
\begin{enumerate}
\item If $\mathbf{A}$ is invertible, $\mathbf{P}$ invertible implies that $\mathbf{C}\mathbf{A}^{-1}\mathbf{B}-\mathbf{D}$ is invertible and
\begin{equation}
\mathbf{A}^\prime=\mathbf{A}^{-1}-\mathbf{A}^{-1}\mathbf{B}\left(\mathbf{C}\mathbf{A}^{-1}\mathbf{B}-\mathbf{D}\right)^{-1}\mathbf{C}\mathbf{A}^{-1},\label{eq:inv1-1}
\end{equation}
\begin{equation}
\mathbf{B}^\prime=\mathbf{A}^{-1}\mathbf{B}\left(\mathbf{C}\mathbf{A}^{-1}\mathbf{B}-\mathbf{D}\right)^{-1},
\end{equation}
\begin{equation}
\mathbf{C}^\prime=\left(\mathbf{C}\mathbf{A}^{-1}\mathbf{B}-\mathbf{D}\right)^{-1}\mathbf{C}\mathbf{A}^{-1},
\end{equation}
\begin{equation}
\mathbf{D}^\prime=-\left(\mathbf{C}\mathbf{A}^{-1}\mathbf{B}-\mathbf{D}\right)^{-1}.\label{eq:inv1-4}
\end{equation}
\item If $\mathbf{D}$ is invertible, $\mathbf{P}$ invertible implies that $\mathbf{B}\mathbf{D}^{-1}\mathbf{C}-\mathbf{A}$ is invertible and
\begin{equation}
\mathbf{A}^\prime=-\left(\mathbf{B}\mathbf{D}^{-1}\mathbf{C}-\mathbf{A}\right)^{-1},\label{eq:inv2-1}
\end{equation}
\begin{equation}
\mathbf{B}^\prime=\left(\mathbf{B}\mathbf{D}^{-1}\mathbf{C}-\mathbf{A}\right)^{-1}\mathbf{B}\mathbf{D}^{-1},
\end{equation}
\begin{equation}
\mathbf{C}^\prime=\mathbf{D}^{-1}\mathbf{C}\left(\mathbf{B}\mathbf{D}^{-1}\mathbf{C}-\mathbf{A}\right)^{-1},
\end{equation}
\begin{equation}
\mathbf{D}^\prime=\mathbf{D}^{-1}-\mathbf{D}^{-1}\mathbf{C}\left(\mathbf{B}\mathbf{D}^{-1}\mathbf{C}-\mathbf{A}\right)^{-1}\mathbf{B}\mathbf{D}^{-1}.\label{eq:inv2-4}
\end{equation}
\item If $\mathbf{A}$ and $\mathbf{D}$ are both invertible, \eqref{eq:inv1-1}-\eqref{eq:inv1-4} and \eqref{eq:inv2-1}-\eqref{eq:inv2-4} are equivalent.
\end{enumerate}
\label{pro:2x2}
\end{proposition}
\begin{proof}
The proposition follows from \cite[Theorem~2.1]{lu02}.
\end{proof}

To rewrite \eqref{eq:v} as a function of $\mathbf{Y}$ by using Proposition~\ref{pro:2x2}, we partition $\mathbf{P}$ and $\mathbf{Y}$ as
\begin{equation}
\mathbf{P}=
\begin{bmatrix}
\mathbf{P}_{11} & \mathbf{P}_{12}\\
\mathbf{P}_{21} & \mathbf{P}_{22}
\end{bmatrix},\;
\mathbf{Y}=
\begin{bmatrix}
\mathbf{Y}_{11} & \mathbf{Y}_{12}\\
\mathbf{Y}_{21} & \mathbf{Y}_{22}
\end{bmatrix},
\end{equation}
where $\mathbf{P}_{11},\mathbf{Y}_{11}\in\mathbb{C}^{N\times N}$, $\mathbf{P}_{12},\mathbf{Y}_{12}\in\mathbb{C}^{N\times M}$, $\mathbf{P}_{21},\mathbf{Y}_{21}\in\mathbb{C}^{M\times N}$, and $\mathbf{P}_{22},\mathbf{Y}_{22}\in\mathbb{C}^{M\times M}$, giving
\begin{equation}
\begin{bmatrix}
\mathbf{P}_{11} & \mathbf{P}_{12}\\
\mathbf{P}_{21} & \mathbf{P}_{22}
\end{bmatrix}
=
\begin{bmatrix}
\mathbf{Y}_{11}/Y_0+\mathbf{I}_N & \mathbf{Y}_{12}/Y_0\\
\mathbf{Y}_{21}/Y_0 & \mathbf{Y}_{22}/Y_0+\mathbf{I}_M
\end{bmatrix},\label{eq:Pblocks}
\end{equation}
as a consequence of \eqref{eq:P}.
Thus, assuming $\mathbf{P}$ to be invertible, with $\mathbf{P}_{11}$ or $\mathbf{P}_{22}$ invertible, the following results hold as a consequence of \eqref{eq:v}:
\begin{enumerate}
\item If $\mathbf{P}_{11}$ is invertible, we have
\begin{multline}
\mathbf{v}_1=\Big(\mathbf{P}_{11}^{-1}-\mathbf{P}_{11}^{-1}\mathbf{P}_{12}\Big.\\
\left.\times\left(\mathbf{P}_{21}\mathbf{P}_{11}^{-1}\mathbf{P}_{12}-\mathbf{P}_{22}\right)^{-1}\mathbf{P}_{21}\mathbf{P}_{11}^{-1}\right)\mathbf{u},\label{eq:v1-1}
\end{multline}
\begin{equation}
\mathbf{v}_2=\left(\left(\mathbf{P}_{21}\mathbf{P}_{11}^{-1}\mathbf{P}_{12}-\mathbf{P}_{22}\right)^{-1}\mathbf{P}_{21}\mathbf{P}_{11}^{-1}\right)\mathbf{u}.\label{eq:v2-1}
\end{equation}
\item If $\mathbf{P}_{22}$ is invertible, we have
\begin{equation}
\mathbf{v}_1=-\left(\mathbf{P}_{12}\mathbf{P}_{22}^{-1}\mathbf{P}_{21}-\mathbf{P}_{11}\right)^{-1}\mathbf{u},\label{eq:v1-2}
\end{equation}
\begin{equation}
\mathbf{v}_2=\left(\mathbf{P}_{22}^{-1}\mathbf{P}_{21}\left(\mathbf{P}_{12}\mathbf{P}_{22}^{-1}\mathbf{P}_{21}-\mathbf{P}_{11}\right)^{-1}\right)\mathbf{u}.\label{eq:v2-2}
\end{equation}
\item If $\mathbf{P}_{11}$ and $\mathbf{P}_{22}$ are both invertible, \eqref{eq:v1-1}-\eqref{eq:v2-1} and \eqref{eq:v1-2}-\eqref{eq:v2-2} are equivalent.
\end{enumerate}
Remarkably, \eqref{eq:v1-1}-\eqref{eq:v2-1} and \eqref{eq:v1-2}-\eqref{eq:v2-2} contain the expressions of the output vectors $\mathbf{v}_1$ and $\mathbf{v}_2$ as a function of the input $\mathbf{u}$ and the blocks of $\mathbf{P}$, which are directly related to the blocks of $\mathbf{Y}$ by \eqref{eq:Pblocks}, when $M>0$.
The expressions of $\mathbf{v}_1$ and $\mathbf{v}_2$ show that a \gls{milac} can perform non-linear computations, such as matrix inversions, while being implemented with a linear microwave network.
This is possible because the input data, i.e., $\mathbf{P}$ and $\mathbf{u}$, is partially encoded into the \gls{milac} structure (specifically, the blocks of $\mathbf{P}$ are encoded into the tunable admittance components), and not only into the input voltage $\mathbf{u}$.
The non-linear mapping between the structure of a linear microwave network and its transfer function has recently received significant attention under the term ``structural non-linearity'' to implement non-linear physical \glspl{nn} based on linear wave systems \cite{mom23,xi24,yi24}.

\begin{figure}[t]
\centering
\includegraphics[width=0.34\textwidth]{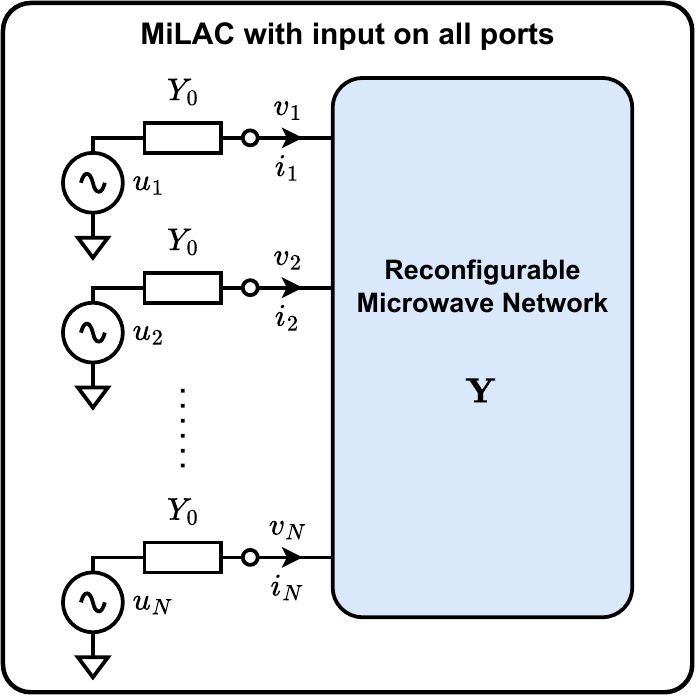}
\caption{Representation of an $N$-port MiLAC with input on all ports.}
\label{fig:ac-N-port}
\end{figure}

Assuming that the tunable components $\{Y_{i,k}\}_{i,k=1}^{P}$ are allowed to take any arbitrary complex value, a \gls{milac} can compute \eqref{eq:v1-1}-\eqref{eq:v2-1} and \eqref{eq:v1-2}-\eqref{eq:v2-2} given any blocks $\mathbf{P}_{11}$, $\mathbf{P}_{12}$, $\mathbf{P}_{21}$, and $\mathbf{P}_{22}$.
This can be achieved by first computing $\mathbf{Y}$ from any arbitrary $\mathbf{P}$ as
\begin{equation}
\mathbf{Y}=Y_0\mathbf{P}-Y_0\mathbf{I}_{P},\label{eq:P-inv}
\end{equation}
which is obtained by inverting \eqref{eq:P}, and then by setting the tunable components $\{Y_{i,k}\}_{i,k=1}^{P}$ as a function of $\mathbf{Y}$ as in \eqref{eq:Yik-component-Y}.
By substituting \eqref{eq:P-inv} into \eqref{eq:Yik-component-Y}, we obtain that the tunable admittance components can be reconfigured as a function of any arbitrary $\mathbf{P}$ as
\begin{equation}
Y_{i,k}=
\begin{cases}
-Y_0\left[\mathbf{P}\right]_{i,k} & i\neq k\\
Y_0\sum_{p=1}^P\left[\mathbf{P}\right]_{p,k}-Y_0 & i=k
\end{cases},\label{eq:Yik-component-P}
\end{equation}
for $i,k=1,\ldots,P$, allowing the \gls{milac} to compute \eqref{eq:v1-1}-\eqref{eq:v2-1} and \eqref{eq:v1-2}-\eqref{eq:v2-2} for any given $\mathbf{P}_{11}$, $\mathbf{P}_{12}$, $\mathbf{P}_{21}$, and $\mathbf{P}_{22}$.

\begin{remark}
The expressions in \eqref{eq:v1-1}-\eqref{eq:v2-1} and \eqref{eq:v1-2}-\eqref{eq:v2-2} would require computationally expensive matrix-matrix products and matrix inversion operations if computed through a digital computer.
However, they can be quickly computed by a \gls{milac} as the \gls{em} signal propagates within the microwave network at a high fraction of the light speed \cite{mcm23}.
In other words, after properly setting the input $\mathbf{u}$ and the tunable components $\{Y_{i,k}\}_{i,k=1}^{P}$, \eqref{eq:v1-1}-\eqref{eq:v2-1} and \eqref{eq:v1-2}-\eqref{eq:v2-2} can be computed with complexity $\mathcal{O}(1)$, independently from the input size $N$ and output size $P=N+M$.
Remarkably, the expressions in \eqref{eq:v1-1}-\eqref{eq:v2-1} and \eqref{eq:v1-2}-\eqref{eq:v2-2} are particularly relevant for computing the \gls{lmmse} estimator and the covariance matrix of its error, as it will be clarified in Section~\ref{sec:lmmse-inv}.
\end{remark}

\begin{table*}[t]
\centering
\caption{Similarity between MiLAC output and LMMSE estimator.}
\begin{tabular}{@{}ccc@{}}
\toprule
MiLAC output & LMMSE estimator and its error covariance matrix &
$\begin{bmatrix}
\mathbf{P}_{11} & \mathbf{P}_{12}\\
\mathbf{P}_{21} & \mathbf{P}_{22}
\end{bmatrix}$\\
\midrule
$\mathbf{v}_2=\left(\left(\mathbf{P}_{21}\mathbf{P}_{11}^{-1}\mathbf{P}_{12}-\mathbf{P}_{22}\right)^{-1}\mathbf{P}_{21}\mathbf{P}_{11}^{-1}\right)\mathbf{u}$
&
$\hat{\mathbf{x}}_{\textnormal{LMMSE},1}
=\left(\left(\mathbf{H}^H\mathbf{C}_{\mathbf{n}}^{-1}\mathbf{H}+\mathbf{C}_{\mathbf{x}}^{-1}\right)^{-1}\mathbf{H}^H\mathbf{C}_{\mathbf{n}}^{-1}\right)\mathbf{y}$
&
$\begin{bmatrix}
\pm\mathbf{C}_\mathbf{n} & \mathbf{H}\\
\mathbf{H}^H & \mp\mathbf{C}_\mathbf{x}^{-1}
\end{bmatrix}$
\\
\midrule
$\mathbf{v}_2=\left(\mathbf{P}_{22}^{-1}\mathbf{P}_{21}\left(\mathbf{P}_{12}\mathbf{P}_{22}^{-1}\mathbf{P}_{21}-\mathbf{P}_{11}\right)^{-1}\right)\mathbf{u}$
&
$\hat{\mathbf{x}}_{\textnormal{LMMSE},2}
=\left(\mathbf{C}_{\mathbf{x}}\mathbf{H}^H\left(\mathbf{H}\mathbf{C}_{\mathbf{x}}\mathbf{H}^H+\mathbf{C}_{\mathbf{n}}\right)^{-1}\right)\mathbf{y}$
&
$\begin{bmatrix}
\pm\mathbf{C}_\mathbf{n} & \mathbf{H}\\
\mathbf{H}^H & \mp\mathbf{C}_\mathbf{x}^{-1}
\end{bmatrix}$
\\
\midrule
$\mathbf{v}_1=\left(\mathbf{P}_{11}^{-1}-\mathbf{P}_{11}^{-1}\mathbf{P}_{12}
\left(\mathbf{P}_{21}\mathbf{P}_{11}^{-1}\mathbf{P}_{12}-\mathbf{P}_{22}\right)^{-1}\mathbf{P}_{21}\mathbf{P}_{11}^{-1}\right)\mathbf{u}$
&
$\mathbf{C}_{\mathbf{e},1}
=\mathbf{C}_{\mathbf{x}}-\mathbf{C}_{\mathbf{x}}\mathbf{H}^H\left(\mathbf{H}\mathbf{C}_{\mathbf{x}}\mathbf{H}^H+\mathbf{C}_{\mathbf{n}}\right)^{-1}\mathbf{H}\mathbf{C}_{\mathbf{x}}$
&
$\begin{bmatrix}
\mathbf{C}_\mathbf{x}^{-1} & \mathbf{H}^H\\
\mathbf{H} & -\mathbf{C}_\mathbf{n}
\end{bmatrix}$
\\
\midrule
$\mathbf{v}_1=-\left(\mathbf{P}_{12}\mathbf{P}_{22}^{-1}\mathbf{P}_{21}-\mathbf{P}_{11}\right)^{-1}\mathbf{u}$
&
$\mathbf{C}_{\mathbf{e},2}
=\left(\mathbf{H}^H\mathbf{C}_{\mathbf{n}}^{-1}\mathbf{H}+\mathbf{C}_{\mathbf{x}}^{-1}\right)^{-1}$
&
$\begin{bmatrix}
\mathbf{C}_\mathbf{x}^{-1} & \mathbf{H}^H\\
\mathbf{H} & -\mathbf{C}_\mathbf{n}
\end{bmatrix}$
\\
\bottomrule
\end{tabular}
\label{tab:MiLAC-LMMSE}
\end{table*}

\subsection{Analysis of a MiLAC with Input on All Ports}
\label{sec:analysis2}

In the case the input is applied on all the ports of the reconfigurable microwave network, i.e., $N=P$ and $M=0$, we have $\mathbf{v}=\mathbf{v}_1$ and $\mathbf{i}=\mathbf{i}_1$, as represented in Fig.~\ref{fig:ac-N-port}.
Thus, this \gls{milac} is characterized by the system
\begin{equation}
\begin{cases}
\mathbf{i}_1&=\mathbf{Y}\mathbf{v}_1\\
\mathbf{i}_1&=Y_0\left(\mathbf{u}-\mathbf{v}_1\right)
\end{cases},\label{eq:system1-M0}
\end{equation}
which yields
\begin{equation}
\mathbf{Y}\mathbf{v}_1=Y_0\left(\mathbf{u}-\mathbf{v}_1\right),
\end{equation}
and allows us to express the output $\mathbf{v}_1$ as
\begin{equation}
\mathbf{v}_1=\mathbf{P}^{-1}\mathbf{u},\label{eq:v-M0}
\end{equation}
where
\begin{equation}
\mathbf{P}=\frac{\mathbf{Y}}{Y_0}+\mathbf{I}_{N},\label{eq:P-M0}
\end{equation}
giving the output vector $\mathbf{v}_1$ as a function of the input $\mathbf{u}$ and the matrix $\mathbf{Y}$, when $M=0$.
Assuming that the tunable components $\{Y_{i,k}\}_{i,k=1}^{N}$ can take any arbitrary complex value, a \gls{milac} can compute \eqref{eq:v-M0} given any matrix $\mathbf{P}$.
To this end, we first compute $\mathbf{Y}$ from any arbitrary $\mathbf{P}$ as
\begin{equation}
\mathbf{Y}=Y_0\mathbf{P}-Y_0\mathbf{I}_{N},\label{eq:P-inv-M0}
\end{equation}
which follows from \eqref{eq:P-M0}, and then we set the tunable components $\{Y_{i,k}\}_{i,k=1}^{N}$ as a function of $\mathbf{Y}$ as in \eqref{eq:Yik-component-Y}.
By substituting \eqref{eq:P-inv-M0} into \eqref{eq:Yik-component-Y}, we obtain the expression of the tunable admittance components as a function of any arbitrary $\mathbf{P}$ as
\begin{equation}
Y_{i,k}=
\begin{cases}
-Y_0\left[\mathbf{P}\right]_{i,k} & i\neq k\\
Y_0\sum_{n=1}^N\left[\mathbf{P}\right]_{n,k}-Y_0 & i=k
\end{cases},\label{eq:Yik-component-P-M0}
\end{equation}
for $i,k=1,\ldots,N$, enabling the \gls{milac} to compute \eqref{eq:v-M0} for any given $\mathbf{P}$.

\begin{remark}
Computing \eqref{eq:v-M0} with a digital computer would involve the expensive computation of the matrix inverse $\mathbf{P}^{-1}$.
However, \eqref{eq:v-M0} can be computed instantly by a \gls{milac}, regardless of the number of inputs $N$ and without requiring any matrix inversion operation.
Remarkably, this enables an efficient computation of the matrix inverse, as it will be detailed in Section~\ref{sec:lmmse-inv}.
\end{remark}


\subsection{Primitive Operations of a MiLAC}

In light of the analysis carried out in Sections~\ref{sec:analysis1} and \ref{sec:analysis2}, a \gls{milac} can be used as a computing device able to perform computing operations in the analog domain.
These operations are computed as the \gls{em} signal propagates within the microwave network at light speed and can be read as output on the network ports.
As a digital computer can compute a set of primitive operations contained in its \gls{isa}, also a \gls{milac} can compute a set of primitive operations which can be sequentially used to evaluate more complicated functions.
We identify three primitive operations that a \gls{milac} with $P$ ports can compute as a function of its input $\mathbf{u}$ applied on $N$ ports and the tunable admittance components $\{Y_{i,k}\}_{i,k=1}^{P}$.
The first primitive, denoted as
\begin{equation}
\mathbf{v}_1=\text{MiLAC}_1\left(\mathbf{u},\{Y_{i,k}\}_{i,k=1}^{P}\right),\label{eq:primitive1}
\end{equation}
returns in $\mathbf{v}_1$ the output on the first $N$ ports of the \gls{milac}.
Mathematically, this primitive is expressed as \eqref{eq:v1-1} or \eqref{eq:v1-2} depending on the invertibility of $\mathbf{P}_{11}$ and $\mathbf{P}_{22}$, where $\mathbf{P}$ is defined as in \eqref{eq:P}.
The second primitive operation, denoted as
\begin{equation}
\mathbf{v}_2=\text{MiLAC}_2\left(\mathbf{u},\{Y_{i,k}\}_{i,k=1}^{P}\right),\label{eq:primitive2}
\end{equation}
gives in $\mathbf{v}_2$ the output on the last $M$ ports of the \gls{milac}, with $M=P-N$.
This primitive is mathematically given by \eqref{eq:v2-1} or \eqref{eq:v2-2} depending on the invertibility of $\mathbf{P}_{11}$ and $\mathbf{P}_{22}$.
Last, the third primitive is denoted as
\begin{equation}
\mathbf{v}=\text{MiLAC}_3\left(\mathbf{u},\{Y_{i,k}\}_{i,k=1}^{P}\right),\label{eq:primitive3}
\end{equation}
and returns the output on all the $P$ ports of the \gls{milac}, i.e., the concatenation of \eqref{eq:primitive1} and \eqref{eq:primitive2}. The primitive \eqref{eq:primitive3} coincides with \eqref{eq:primitive1} when the input is applied on all the ports, i.e., $N=P$ and $M=0$.
The sequential use of different primitive operations computed in the analog domain can enable the efficient analog computation of complex algorithms.
Nevertheless, it remains less investigated in the literature \cite{amb10}-\cite{wan23}.
In the following, we show how to exploit these three primitives together with low-complexity digital operations to efficiently compute functions with practical interest, namely the \gls{lmmse} estimator, the matrix inversion, and the Kalman filter.

\section{Analog Computing for the LMMSE Estimator and Matrix Inversion}
\label{sec:lmmse-inv}

We have introduced a \gls{milac} as an analog computer that linearly processes microwave signals and we characterized its output as a function of its input and microwave network admittance matrix.
In this section, we show that a \gls{milac} can be used to compute the \gls{lmmse} estimator for linear observation processes and invert a matrix with significantly reduced computational complexity.

\subsection{LMMSE Estimator for Linear Observation Processes}

Consider a linear observation process in which the known random vector $\mathbf{y}\in\mathbb{C}^{Y\times 1}$, also denoted as the observation vector, is given by
\begin{equation}
\mathbf{y}=\mathbf{H}\mathbf{x}+\mathbf{n},\label{eq:lop}
\end{equation}
where $\mathbf{H}\in\mathbb{C}^{Y\times X}$ is a known constant matrix, $\mathbf{x}\in\mathbb{C}^{X\times 1}$ is the unknown random vector with mean $\bar{\mathbf{x}}=\mathbf{0}_{X\times 1}$ and covariance matrix $\mathbf{C}_\mathbf{x}$, and $\mathbf{n}\in\mathbb{C}^{Y\times 1}$ is the random noise vector with mean $\bar{\mathbf{n}}=\mathbf{0}_{Y\times 1}$ and covariance matrix $\mathbf{C}_\mathbf{n}$.
We assume that $\mathbf{C}_{\mathbf{x}}$ and $\mathbf{C}_{\mathbf{n}}$ are known and invertible, and that the cross-covariance matrix of $\mathbf{x}$ and $\mathbf{n}$ is $\mathbf{C}_{\mathbf{x}\mathbf{n}}=\mathbf{0}_{X\times Y}$.
According to estimation theory, an estimator $\hat{\mathbf{x}}\in\mathbb{C}^{X\times 1}$ of $\mathbf{x}$ is any function of $\mathbf{y}$, its estimation error vector writes as $\mathbf{e}=\hat{\mathbf{x}}-\mathbf{x}$, and the \gls{mse} is given by $\text{E}[\Vert\hat{\mathbf{x}}-\mathbf{x}\Vert_2^{2}]$ \cite[Chapter 2]{kay93}.
The \gls{lmmse} estimator $\hat{\mathbf{x}}_{\text{LMMSE}}$ of $\mathbf{x}$ is defined as the linear estimator minimizing the \gls{mse}, i.e.,
\begin{equation}
\hat{\mathbf{x}}_{\text{LMMSE}}=\underset{\hat{\mathbf{x}}}{\text{min}}\;
\text{E}\left[\left\Vert\hat{\mathbf{x}}-\mathbf{x}\right\Vert_2^{2}\right]
\;\mathsf{\mathrm{s.t.}}\;
\hat{\mathbf{x}}=\mathbf{W}\mathbf{y}+\mathbf{b},
\end{equation}
where $\mathbf{W}\in\mathbb{C}^{X\times Y}$ and $\mathbf{b}\in\mathbb{C}^{X\times 1}$.
The expression of the \gls{lmmse} estimator for a linear observation process is available in closed form, as given by the following proposition.

\begin{proposition}
Considering the observation process in \eqref{eq:lop}, the \gls{lmmse} estimator $\hat{\mathbf{x}}_{\text{LMMSE}}$ of $\mathbf{x}$ is equivalently given by
\begin{align}
\hat{\mathbf{x}}_{\textnormal{LMMSE},1}
&=\left(\left(\mathbf{H}^H\mathbf{C}_{\mathbf{n}}^{-1}\mathbf{H}+\mathbf{C}_{\mathbf{x}}^{-1}\right)^{-1}\mathbf{H}^H\mathbf{C}_{\mathbf{n}}^{-1}\right)\mathbf{y},\label{eq:LMMSE1}\\
\hat{\mathbf{x}}_{\textnormal{LMMSE},2}
&=\left(\mathbf{C}_{\mathbf{x}}\mathbf{H}^H\left(\mathbf{H}\mathbf{C}_{\mathbf{x}}\mathbf{H}^H+\mathbf{C}_{\mathbf{n}}\right)^{-1}\right)\mathbf{y}\label{eq:LMMSE2},
\end{align}
where it holds that $\hat{\mathbf{x}}_{\textnormal{LMMSE},1}=\hat{\mathbf{x}}_{\textnormal{LMMSE},2}$.
Furthermore, the \gls{lmmse} estimation error vector $\mathbf{e}=\hat{\mathbf{x}}_{\text{LMMSE}}-\mathbf{x}$ has mean $\bar{\mathbf{e}}=\mathbf{0}_{X\times 1}$ and covariance matrix $\mathbf{C}_{\mathbf{e}}$ equivalently given by
\begin{align}       
\mathbf{C}_{\mathbf{e},1}
&=\mathbf{C}_{\mathbf{x}}-\mathbf{C}_{\mathbf{x}}\mathbf{H}^H\left(\mathbf{H}\mathbf{C}_{\mathbf{x}}\mathbf{H}^H+\mathbf{C}_{\mathbf{n}}\right)^{-1}\mathbf{H}\mathbf{C}_{\mathbf{x}},\label{eq:Ce1}\\
\mathbf{C}_{\mathbf{e},2}
&=\left(\mathbf{H}^H\mathbf{C}_{\mathbf{n}}^{-1}\mathbf{H}+\mathbf{C}_{\mathbf{x}}^{-1}\right)^{-1},\label{eq:Ce2}
\end{align}
where we have $\mathbf{C}_{\mathbf{e},1}=\mathbf{C}_{\mathbf{e},2}$.
\label{pro:lmmse}
\end{proposition}
\begin{proof}
The proposition follows from \cite[Theorem~12.1]{kay93}.
\end{proof}

Interestingly, there is a remarkable similarity between the \gls{lmmse} expressions in \eqref{eq:LMMSE1}-\eqref{eq:Ce2} and the expressions characterizing the output of a \gls{milac}, as highlighted in Tab.~\ref{tab:MiLAC-LMMSE}.
Specifically, we notice a similarity between the \gls{lmmse} estimator expressions in \eqref{eq:LMMSE1} and \eqref{eq:LMMSE2} and the expressions of $\mathbf{v}_2$ in \eqref{eq:v2-1} and \eqref{eq:v2-2}, respectively.
Likewise, \eqref{eq:Ce1} and \eqref{eq:Ce2} have the same structure as the matrix multiplying $\mathbf{u}$ in \eqref{eq:v1-1} and \eqref{eq:v1-2}, respectively.
Thus, we can efficiently compute the \gls{lmmse} estimator and the covariance matrix of its error with a \gls{milac}, as specified in the following.

\subsection{Computing the LMMSE Estimator with a MiLAC}
\label{sec:lmmse-milac}

\begin{algorithm}[t]
\begin{algorithmic}[1]
\REQUIRE $\mathbf{y}$, $\mathbf{H}$, $\mathbf{C}_{\mathbf{x}}^{-1}$, and $\mathbf{C}_{\mathbf{n}}$.
\ENSURE $\hat{\mathbf{x}}_{\text{LMMSE}}$.
\STATE{
Set $\{Y_{i,k}\}_{i,k=1}^{P}$ by \eqref{eq:Yik-component-P} with $\mathbf{P}=\mathbf{P}_{\text{LMMSE}}$.}
\STATE{$\hat{\mathbf{x}}_{\text{LMMSE}}=\text{MiLAC}_2(\mathbf{y},\{Y_{i,k}\}_{i,k=1}^{P})$.}
\end{algorithmic}
\caption{LMMSE estimator with MiLAC}
\label{alg:LMMSE}
\end{algorithm}

\begin{algorithm}[t]
\begin{algorithmic}[1]
\REQUIRE $\mathbf{H}$, $\mathbf{C}_{\mathbf{x}}^{-1}$, and $\mathbf{C}_{\mathbf{n}}$.
\ENSURE $\mathbf{C}_{\mathbf{e}}$.
\STATE{Set $\{Y_{i,k}\}_{i,k=1}^{P}$ by \eqref{eq:Yik-component-P} with $\mathbf{P}=\mathbf{P}_{\text{Cov}}$.}
\FOR{$n=1\textbf{ to }N$}
\STATE{$[\mathbf{C}_{\mathbf{e}}]_{:,n}=\text{MiLAC}_1([\mathbf{I}_{N}]_{:,n},\{Y_{i,k}\}_{i,k=1}^{P})$.}
\ENDFOR
\end{algorithmic}
\caption{Covariance matrix of the LMMSE estimator error with MiLAC}
\label{alg:covariance}
\end{algorithm}

\begin{algorithm}[t!]
\begin{algorithmic}[1]
\REQUIRE $\mathbf{P}$.
\ENSURE $\mathbf{P}^{-1}$.
\STATE{Set $\{Y_{i,k}\}_{i,k=1}^{N}$ by \eqref{eq:Yik-component-P-M0}.}
\FOR{$n=1\textbf{ to }N$}
\STATE{$[\mathbf{P}^{-1}]_{:,n}=\text{MiLAC}_3([\mathbf{I}_{N}]_{:,n},\{Y_{i,k}\}_{i,k=1}^{N})$.}
\ENDFOR
\end{algorithmic}
\caption{Matrix inversion with MiLAC}
\label{alg:inversion}
\end{algorithm}

\begin{figure}[t]
\centering
\includegraphics[width=0.42\textwidth]{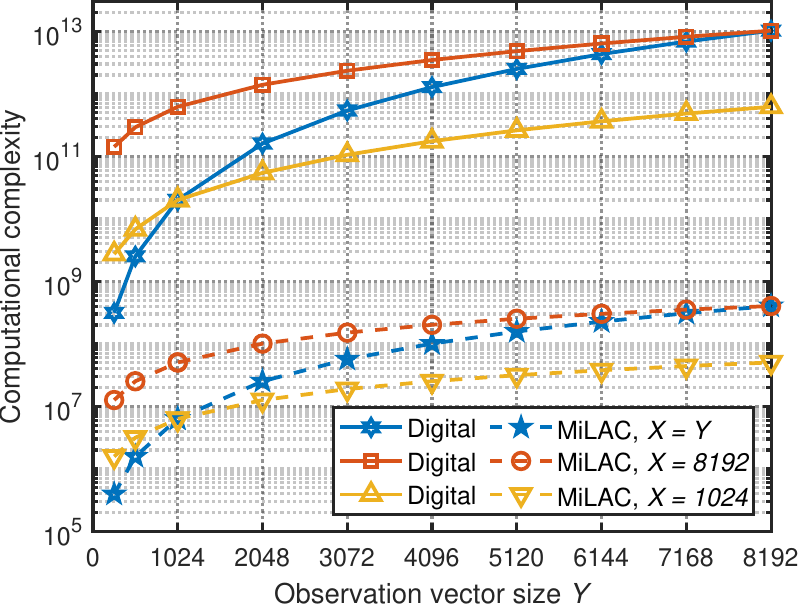}
\caption{Computational complexity of computing the LMMSE estimator given an observation vector with size $Y$, for various sizes $X$ of the unknown vector.}
\label{fig:complexity-LMMSE}
\end{figure}

\begin{figure}[t]
\centering
\includegraphics[width=0.42\textwidth]{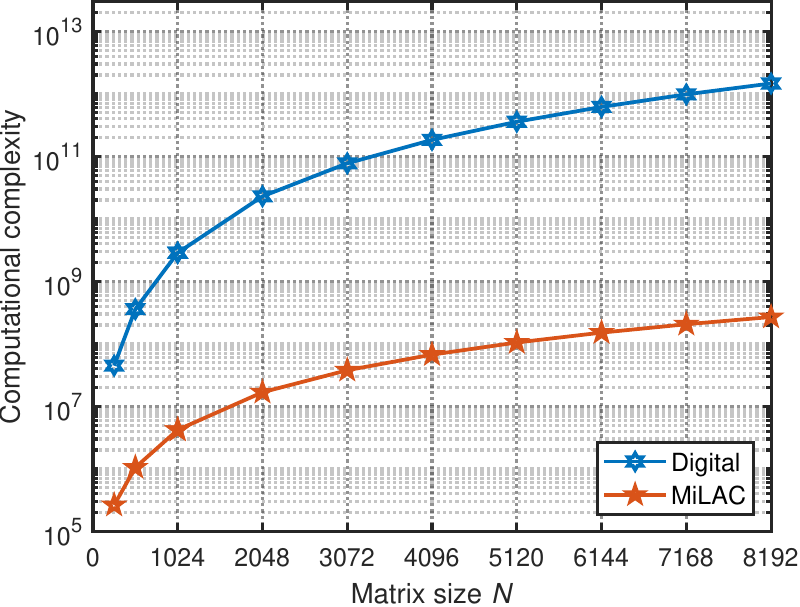}
\caption{Computational complexity of inverting a matrix with size $N$.}
\label{fig:complexity-inversion}
\end{figure}

The expressions of the \gls{lmmse} estimator in \eqref{eq:LMMSE1} and \eqref{eq:LMMSE2} can be obtained on the output $\mathbf{v}_2$ of a \gls{milac}, i.e., $\mathbf{v}_{2}=\hat{\mathbf{x}}_{\text{LMMSE},1}$ and $\mathbf{v}_{2}=\hat{\mathbf{x}}_{\text{LMMSE},2}$, by exploiting \eqref{eq:v2-1} and \eqref{eq:v2-2}, respectively.
To this end, the input vector of the \gls{milac} is set as $\mathbf{u}=\mathbf{y}$, and the tunable admittance components $\{Y_{i,k}\}_{i,k=1}^{P}$ are set according to \eqref{eq:Yik-component-P}, with $\mathbf{P}=\mathbf{P}_{\text{LMMSE}}$, where
\begin{equation}
\mathbf{P}_{\text{LMMSE}}=
\begin{bmatrix}
\pm\mathbf{C}_\mathbf{n} & \mathbf{H}\\
\mathbf{H}^H & \mp\mathbf{C}_\mathbf{x}^{-1}
\end{bmatrix}\label{eq:P-lmmse}
\end{equation}
are two values of $\mathbf{P}_{\text{LMMSE}}$ giving the desired output.
Both values of $\mathbf{P}_{\text{LMMSE}}$ included in \eqref{eq:P-lmmse} result in the desired \gls{lmmse} estimator expression, as it can be shown by analyzing the expression of the \gls{milac} output in \eqref{eq:v2-1} and \eqref{eq:v2-2}.
The steps necessary to compute the \gls{lmmse} estimator in the analog domain with a \gls{milac} are summarized in Alg.~\ref{alg:LMMSE}, where the \gls{milac} has $N=Y$ input ports and $M=X$ output ports with no input, namely a total of $P=X+Y$ ports.
Note that analog computers have been proposed to compute generalized regression, ridge regression, and linear regression through resistive memory arrays \cite{sun22,man23,man22}.
Interestingly, all these operations can be seen as special cases of the \gls{lmmse} estimator, confirming the generality of our model.

\begin{table*}[t]
\centering
\caption{Similarity between MiLAC output and Kalman filter.}
\begin{tabular}{@{}cc@{}}
\toprule
MiLAC output & Kalman filter equations\\
\midrule
$\mathbf{v}_1=-\left(\mathbf{P}_{12}\mathbf{P}_{22}^{-1}\mathbf{P}_{21}-\mathbf{P}_{11}\right)^{-1}\mathbf{u}$
&
$\mathbf{R}_{t\mid t-1}^{-1}=\left(\mathbf{A}_{t}\mathbf{R}_{t-1\mid t-1}\mathbf{A}_{t}^{H}+\mathbf{M}\right)^{-1}$
\\
\midrule
$\mathbf{v}_2=\left(\mathbf{P}_{22}^{-1}\mathbf{P}_{21}\left(\mathbf{P}_{12}\mathbf{P}_{22}^{-1}\mathbf{P}_{21}-\mathbf{P}_{11}\right)^{-1}\right)\mathbf{u}$
&
$\hat{\mathbf{x}}_{t\mid t}-\hat{\mathbf{x}}_{t\mid t-1}
=\left(\mathbf{R}_{t\mid t-1}\mathbf{H}_{t}^{H}\left(\mathbf{H}_{t}{\mathbf{R}}_{t\mid t-1}\mathbf{H}_{t}^{H}+\mathbf{N}\right)^{-1}\right)\left(\mathbf{y}_{t}-\mathbf{H}_{t}\hat{\mathbf{x}}_{t\mid t-1}\right)$
\\
\midrule
$\mathbf{v}_1=\left(\mathbf{P}_{11}^{-1}-\mathbf{P}_{11}^{-1}\mathbf{P}_{12}
\left(\mathbf{P}_{21}\mathbf{P}_{11}^{-1}\mathbf{P}_{12}-\mathbf{P}_{22}\right)^{-1}\mathbf{P}_{21}\mathbf{P}_{11}^{-1}\right)\mathbf{u}$
&
$\mathbf{R}_{t\mid t}
=\mathbf{R}_{t\mid t-1}-\mathbf{R}_{t\mid t-1}\mathbf{H}_{t}^{H}\left(\mathbf{H}_{t}{\mathbf{R}}_{t\mid t-1}\mathbf{H}_{t}^{H}+\mathbf{N}\right)^{-1}\mathbf{H}_{t}\mathbf{R}_{t\mid t-1}$
\\
\bottomrule
\end{tabular}
\label{tab:MiLAC-Kalman}
\end{table*}

\begin{remark}
A \gls{milac} can compute the \gls{lmmse} estimator with significantly reduced computational complexity compared to a digital computer, as no matrix-matrix product or matrix inversion operation is required.
While a \gls{milac} computes the \gls{lmmse} estimator entirely in the analog domain requiring no digital operation, computing the values of its tunable admittance components requires a digital computer.
Thus, the computational complexity of computing the \gls{lmmse} estimator with a \gls{milac} is given by the complexity of digitally computing the tunable admittance components in \eqref{eq:Yik-component-P}, in terms of the number of arithmetic real operations required.
By exploiting the fact that $\mathbf{P}$ in \eqref{eq:P-lmmse} is Hermitian and precomputing offline the computations involving $\mathbf{C}_{\mathbf{x}}$ and $\mathbf{C}_{\mathbf{n}}$, \eqref{eq:Yik-component-P} require approximately $2XY$ and $4XY$ real operations to be digitally computed for the two cases $i\neq k$ and $i=k$, respectively, giving $6XY$ real operations in total.
Conversely, the number of real operations required to compute \eqref{eq:LMMSE1} and \eqref{eq:LMMSE2} through digital computing is $8XY^2+8X^2Y+8X^3/3$ and $8X^2Y+8XY^2+8Y^3/3$, respectively, due to the two matrix-matrix products and the matrix inversion, whose complexity is assessed in the Appendix.
As \eqref{eq:LMMSE1} and \eqref{eq:LMMSE2} are equivalent, the complexity of computing the \gls{lmmse} estimator digitally is $\mathrm{min}\{8XY^2+8X^2Y+8X^3/3,8X^2Y+8XY^2+8Y^3/3\}$, i.e., $8(XY^2+X^2Y+\mathrm{min}\{X^3,Y^3\}/3)$.
The complexity of computing the \gls{lmmse} estimator with a \gls{milac} and by digital computing are compared in Fig.~\ref{fig:complexity-LMMSE}, showing a gain of $2.5\times 10^4$ times in the case $X=Y=8192$.
\label{rem:lmmse}
\end{remark}

\subsection{Computing the Covariance Matrix of the LMMSE Estimator Error with a MiLAC}

In addition to the \gls{lmmse} estimator, a \gls{milac} can also efficiently compute the covariance matrix of the \gls{lmmse} estimation error in \eqref{eq:Ce1} and \eqref{eq:Ce2} by exploiting \eqref{eq:v1-1} and \eqref{eq:v1-2}, respectively.
To this end, the tunable admittance components $\{Y_{i,k}\}_{i,k=1}^{P}$ are set as in \eqref{eq:Yik-component-P} with $\mathbf{P}=\mathbf{P}_{\text{Cov}}$ where
\begin{equation}
\mathbf{P}_{\text{Cov}}=
\begin{bmatrix}
\mathbf{C}_{\mathbf{x}}^{-1} & \mathbf{H}^H\\
\mathbf{H} & -\mathbf{C}_{\mathbf{n}}
\end{bmatrix},
\end{equation}
requiring a \gls{milac} with $N=X$ input ports and $M=Y$ output ports with no input, i.e., a total of $P=X+Y$ ports.
Thus, by setting the input vector as $\mathbf{u}=[\mathbf{I}_N]_{:,n}$, the $n$th column of the desired covariance matrix $\mathbf{C}_{\mathbf{e}}$ is returned on the output $\mathbf{v}_{1}$ following \eqref{eq:v1-1} and \eqref{eq:v1-2}, i.e., $\mathbf{v}_{1}=[\mathbf{C}_{\mathbf{e},1}]_{:,n}$ or, equivalently, $\mathbf{v}_{1}=[\mathbf{C}_{\mathbf{e},2}]_{:,n}$.
By iteratively considering all the input vectors $\mathbf{u}=[\mathbf{I}_N]_{:,n}$, for $n=1,\ldots,N$, all the $N$ columns of the matrix $\mathbf{C}_{\mathbf{e}}$ can be obtained, as summarized in Alg.~\ref{alg:covariance}.

\begin{remark}
By using a \gls{milac} to compute the covariance matrix of the \gls{lmmse} estimation error, we can achieve the same benefit in terms of computational complexity as observed for the computation of the \gls{lmmse} estimator in Remark~\ref{rem:lmmse}.
As a \gls{milac} computes the desired covariance matrix entirely in the analog domain requiring no digital operation, the required computational complexity is driven by the complexity of digitally computing the values of its tunable admittance components as in \eqref{eq:Yik-component-P}.
Specifically, the covariance matrix can be computed by a \gls{milac} with only $6XY$ real operations (digitally computed), while it would require $8(XY^2+X^2Y+\mathrm{min}\{X^3,Y^3\}/3)$ real operations to compute \eqref{eq:Ce1} or \eqref{eq:Ce2} by a digital computer.
\end{remark}

\subsection{Inverting a Matrix with a MiLAC}

A \gls{milac} with input on all ports ($N=P$ and $M=0$), whose input-output relationship is analyzed in Section~\ref{sec:analysis2}, can be used to efficiently compute the inverse of an arbitrary matrix $\mathbf{P}\in\mathbb{C}^{N\times N}$.
To this end, the tunable admittance components of the \gls{milac} $\{Y_{i,k}\}_{i,k=1}^{N}$ need to be set according to \eqref{eq:Yik-component-P-M0}, such that the output vector is $\mathbf{v}_1=\mathbf{P}^{-1}\mathbf{u}$, as discussed in Section~\ref{sec:analysis2}.
Thus, by setting the input vector as $\mathbf{u}=[\mathbf{I}_N]_{:,n}$, the $n$th column of $\mathbf{P}^{-1}$ is given on the output $\mathbf{v}_{1}$, i.e., $\mathbf{v}_{1}=[\mathbf{P}^{-1}]_{:,n}$.
To obtain all the $N$ columns of $\mathbf{P}^{-1}$, we can iteratively consider all the input vectors $\mathbf{u}=[\mathbf{I}_N]_{:,n}$, for $n=1,\ldots,N$, as reported in Alg.~\ref{alg:inversion}.
Note that analog computers specifically designed to compute the matrix inversion have appeared in the literature \cite{wu14,moh19,sun22}, confirming the practicability of such an approach to significantly decrease the complexity of the matrix inversion operation.
Since the prototypes in \cite{wu14,moh19} operate with \gls{rf} signals, they can be regarded as a specific implementation of a \gls{milac} having a fixed (non-reconfigurable) admittance matrix.
Besides, the resistive memory array used in \cite{sun22} cannot be regarded as an implementation of \gls{milac} as it does not compute with \gls{rf} signals.
However, it shares similar capabilities with a \gls{milac} as they are governed by the same \gls{em} laws, such as Ohm's law.

\begin{remark}
The computational complexity of this matrix inversion approach is driven by the number of operations required to compute the tunable admittance components according to \eqref{eq:Yik-component-P-M0}, given by approximately $4N^2$ real operations in total.
Once the tunable admittance components are reconfigured, a \gls{milac} can compute the matrix inverse with only $N$ measurements, where each measurement returns a column of the matrix inverse.
Thus, a \gls{milac} can compute the inverse of an arbitrary $N\times N$ matrix with complexity $\mathcal{O}(N^2)$, which is significantly less than the complexity of matrix inversion with digital computing given by $8N^3/3$, as discussed in the Appendix.
The complexity of inverting a matrix with a \gls{milac} and by digital computing are compared in Fig.~\ref{fig:complexity-inversion}, showing a gain of $5.5\times 10^3$ times in the case $N=8192$.
\end{remark}

\section{Analog Computing for the Kalman Filter}
\label{sec:kalman}

We have discussed how a \gls{milac} can compute the \gls{lmmse} estimator and perform matrix inversion with highly reduced computational complexity compared to digital computing.
In this section, we show that a \gls{milac} can also efficiently perform Kalman filtering, which remains unexplored in the context of analog computing \cite{amb10}-\cite{wan23}, and is a more general operation than the \gls{lmmse} estimator \cite{kal60}.

\begin{algorithm}[t]
\begin{algorithmic}[1]
\REQUIRE $\mathbf{A}_{t}$, $\mathbf{H}_{t}$, $\mathbf{M}$, $\mathbf{N}$, $\hat{\mathbf{x}}_{t-1\mid t-1}$, and $\mathbf{R}_{t-1\mid t-1}$.
\ENSURE $\hat{\mathbf{x}}_{t\mid t}$ and $\mathbf{R}_{t\mid t}$.
\STATEx{\textbf{Prediction:}}
\STATE{$\hat{\mathbf{x}}_{t\mid t-1}=\mathbf{A}_{t}{\hat{\mathbf{x}}}_{t-1\mid t-1}$.}
\STATE{$\mathbf{R}_{t\mid t-1}=\mathbf{A}_{t}\mathbf{R}_{t-1\mid t-1}\mathbf{A}_{t}^{H}+\mathbf{M}$.}
\STATEx{\textbf{Correction:}}
\STATE{$\mathbf{K}_{t}=\mathbf{R}_{t\mid t-1}\mathbf{H}_{t}^{H}\left(\mathbf{H}_{t}{\mathbf{R}}_{t\mid t-1}\mathbf{H}_{t}^{H}+\mathbf{N}\right)^{-1}$.}
\STATE{$\hat{\mathbf{x}}_{t\mid t}
=\hat{\mathbf{x}}_{t\mid t-1}
+\mathbf{K}_{t}\left(\mathbf{y}_{t}-\mathbf{H}_{t}\hat{\mathbf{x}}_{t\mid t-1}\right)$.}
\STATE{$\mathbf{R}_{t\mid t}
=\left(\mathbf{I}-\mathbf{K}_{t}\mathbf{H}_{t}\right)\mathbf{R}_{t\mid t-1}$.}
\end{algorithmic}
\caption{Kalman filter}
\label{alg:Kalman-digital}
\end{algorithm}

\begin{algorithm}[t]
\begin{algorithmic}[1]
\REQUIRE $\mathbf{A}_{t}$, $\mathbf{H}_{t}$, $\mathbf{M}$, $\mathbf{N}$, $\hat{\mathbf{x}}_{t-1\mid t-1}$, and $\mathbf{R}_{t-1\mid t-1}^{-1}$.
\ENSURE $\hat{\mathbf{x}}_{t\mid t}$ and $\mathbf{R}_{t\mid t}^{-1}$.
\STATEx{\textbf{Prediction:}}
\STATE{$\hat{\mathbf{x}}_{t\mid t-1}=\mathbf{A}_{t}{\hat{\mathbf{x}}}_{t-1\mid t-1}$.}
\STATE{Set $\mathbf{R}_{t\mid t-1}^{-1}$ by Alg.~2 with Input: $\mathbf{A}_{t}^H$, $\mathbf{M}$, and $\mathbf{R}_{t-1\mid t-1}^{-1}$.}
\STATEx{\textbf{Correction:}}
\STATE{$\bar{\mathbf{y}}_{t}=\mathbf{y}_{t}-\mathbf{H}_{t}\hat{\mathbf{x}}_{t\mid t-1}$.}
\STATE{Set $\bar{\mathbf{x}}_{t\mid t}$ by Alg.~1 with Input: $\bar{\mathbf{y}}_{t}$, $\mathbf{H}_{t}^H$, $\mathbf{N}$, and $\mathbf{R}_{t\mid t-1}^{-1}$.}
\STATE{$\hat{\mathbf{x}}_{t\mid t}=\bar{\mathbf{x}}_{t\mid t}+\hat{\mathbf{x}}_{t\mid t-1}$.}
\STATE{Set $\mathbf{R}_{t\mid t}$ by Alg.~2 with Input: $\mathbf{A}_{t}^H$, $\mathbf{R}_{t-1\mid t-1}^{-1}$, and $\mathbf{N}$.}
\STATE{Set $\mathbf{R}_{t\mid t}^{-1}$ by Alg.~3 with Input: $\mathbf{R}_{t\mid t}$.}
\end{algorithmic}
\caption{Kalman filter with MiLAC}
\label{alg:Kalman-analog}
\end{algorithm}

\subsection{Kalman Filter}

Consider a discrete-time dynamical process whose state at time $t$, denoted as $\mathbf{x}_{t}\in\mathbb{C}^{X\times 1}$, is a linear function of the state at time $t-1$, denoted as $\mathbf{x}_{t-1}\in\mathbb{C}^{X\times 1}$, given by
\begin{equation}
\mathbf{x}_{t}=\mathbf{A}_{t}\mathbf{x}_{t-1}+\mathbf{m}_{t},
\end{equation}
where $\mathbf{A}_{t}\in\mathbb{C}^{X\times X}$ is a known state transition matrix and $\mathbf{m}_{t}\in\mathbb{C}^{X\times 1}$ is the random state noise with mean $\bar{\mathbf{m}}=\mathbf{0}_{X\times 1}$ and covariance matrix $\mathbf{M}$.
At time $t$, an observation $\mathbf{y}_{t}\in\mathbb{C}^{Y\times 1}$ of the state $\mathbf{x}_{t}$ is acquired according to
\begin{equation}
\mathbf{y}_{t}=\mathbf{H}_{t}\mathbf{x}_{t}+\mathbf{n}_{t},
\end{equation}
where $\mathbf{H}_{t}\in\mathbb{C}^{Y\times X}$ is a known observation matrix and $\mathbf{n}_{t}\in\mathbb{C}^{Y\times 1}$ is the random observation noise with mean $\bar{\mathbf{n}}=\mathbf{0}_{Y\times 1}$ and covariance matrix $\mathbf{N}$\footnote{For simplicity, we use $X$ and $Y$ to denote the dimensionalities in both the \gls{lmmse} estimator and the Kalman filter, relying on context to distinguish their meaning.}.
The goal of the Kalman filter is to provide at time $t$ an estimate of the state $\mathbf{x}_{t}$ as a function of the estimate of the state $\mathbf{x}_{t-1}$ at time $t-1$ and the noisy observation $\mathbf{y}_{t}$ \cite[Chapter 13]{kay93}, \cite{bis01}.
To this end, the Kalman filter operates in two phases, namely ``prediction'' and ``correction'' \cite[Chapter 13]{kay93}, \cite{bis01}.

In the ``prediction'' phase, an a priori estimate of the current state is computed depending on the previous state estimate.
The a priori estimate of $\mathbf{x}_{t}$, denoted as $\hat{\mathbf{x}}_{t\mid t-1}$, is given by
\begin{equation}
\hat{\mathbf{x}}_{t\mid t-1}=\mathbf{A}_{t}\hat{\mathbf{x}}_{t-1\mid t-1},\label{eq:Kalman1}
\end{equation}
where $\hat{\mathbf{x}}_{t-1\mid t-1}$ is the a posteriori estimate of $\mathbf{x}_{t-1}$.
In addition, the Kalman filter computes the a priori estimate error covariance matrix $\mathbf{R}_{t\mid t-1}\in\mathbb{C}^{X\times X}$ as
\begin{equation}
\mathbf{R}_{t\mid t-1}=\mathbf{A}_{t}\mathbf{R}_{t-1\mid t-1}\mathbf{A}_{t}^{H}+\mathbf{M},\label{eq:Kalman2}
\end{equation}
as a function of the a posteriori estimate error covariance matrix at time $t-1$ denoted as $\mathbf{R}_{t-1\mid t-1}$.

\begin{figure}[t]
\centering
\includegraphics[width=0.42\textwidth]{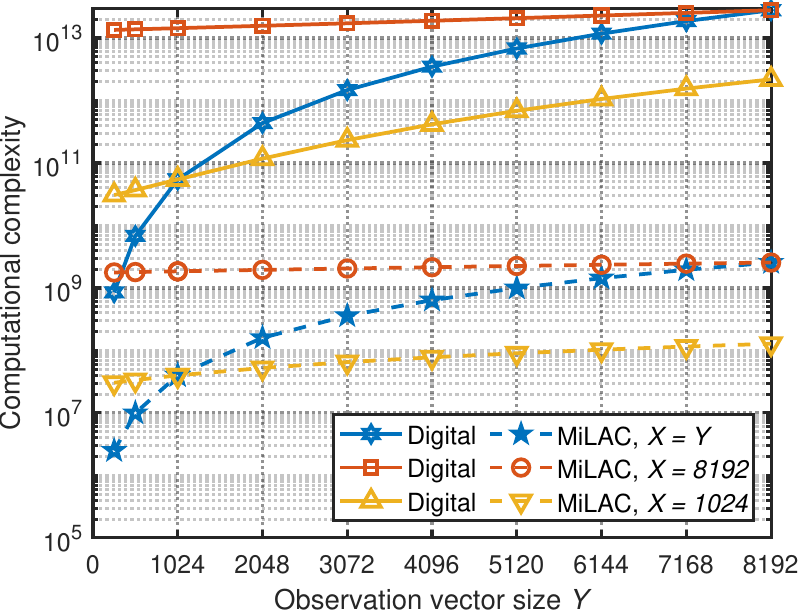}
\caption{Computational complexity of computing the Kalman filter given an observation vector with size $Y$, for various sizes $X$ of the state vector.}
\label{fig:complexity-Kalman}
\end{figure}

In the ``correction'' phase, the current state estimate is refined based on the current a priori estimate and the current observation.
This refined estimate is named the a posteriori estimate of the state.
After computing the so-called Kalman gain $\mathbf{K}_{t}\in\mathbb{C}^{X\times Y}$ as
\begin{equation}
\mathbf{K}_{t}=\mathbf{R}_{t\mid t-1}\mathbf{H}_{t}^{H}\left(\mathbf{H}_{t}{\mathbf{R}}_{t\mid t-1}\mathbf{H}_{t}^{H}+\mathbf{N}\right)^{-1},\label{eq:Kalman3}
\end{equation}
the a posteriori estimate of $\mathbf{x}_{t}$, denoted as $\hat{\mathbf{x}}_{t\mid t}$, is given by
\begin{equation}
\hat{\mathbf{x}}_{t\mid t}
=\hat{\mathbf{x}}_{t\mid t-1}+\mathbf{K}_{t}\left(\mathbf{y}_{t}-\mathbf{H}_{t}\hat{\mathbf{x}}_{t\mid t-1}\right),\label{eq:Kalman4}
\end{equation}
as a function of the a priori estimate $\hat{\mathbf{x}}_{t\mid t-1}$ and the observation $\mathbf{y}_{t}$.
Besides, the a posteriori estimate error covariance matrix is also computed as
\begin{equation}
\mathbf{R}_{t\mid t}=\left(\mathbf{I}-\mathbf{K}_{t}\mathbf{H}_{t}\right){\mathbf{R}}_{t\mid t-1},\label{eq:Kalman5}
\end{equation}
refining the a priori estimate error covariance matrix $\mathbf{R}_{t\mid t-1}$.
The two phases ``prediction'' and ``correction'' are sequentially repeated at each time step, as reported in Alg.~\ref{alg:Kalman-digital}\footnote{The provided complex-valued Kalman filter implicitly assumes circularly symmetric initial state, state noise, and observation noise. To treat non-circularly symmetric random variables, augmented complex-valued Kalman filters have been proposed \cite{din12}.}.
This process ensures that the Kalman filter output $\hat{\mathbf{x}}_{t\mid t}$ is the \gls{lmmse} estimator of $\mathbf{x}_{t}$ \cite[Chapter 13]{kay93}.

We observe an interesting similarity between the operations involved in the Kamlan filter \eqref{eq:Kalman1}-\eqref{eq:Kalman5} and the expressions computable with a \gls{milac}.
Specifically, we highlight three similarities, as visualized in Tab.~\ref{tab:MiLAC-Kalman}.
First, the inverse of \eqref{eq:Kalman2} has the same structure as the matrix multiplying $\mathbf{u}$ in \eqref{eq:v1-2}.
Second, $\hat{\mathbf{x}}_{t\mid t}-\hat{\mathbf{x}}_{t\mid t-1}$ given by \eqref{eq:Kalman4} has the same structure as \eqref{eq:v2-2}.
Third, \eqref{eq:Kalman5} has the same structure as the matrix multiplying $\mathbf{u}$ in \eqref{eq:v1-1}.
Thus, the Kalman filter can be efficiently computed with a \gls{milac}, as discussed in the following.

\subsection{Computing the Kalman Filter with a MiLAC}

The Kalman filter in Alg.~\ref{alg:Kalman-digital} can be computed by performing the most computationally expensive steps in the analog domain through a \gls{milac}.
This is shown in Alg.~\ref{alg:Kalman-analog}, which takes in input the matrix inverse $\mathbf{R}_{t-1\mid t-1}^{-1}$ and return the matrix inverse $\mathbf{R}_{t\mid t}^{-1}$, differently from Alg.~\ref{alg:Kalman-digital}, but without altering the logic of the Kalman filter.
In the ``prediction'' phase, Step~1 in Alg.~\ref{alg:Kalman-digital} has low computational complexity, i.e., $8X^2$ real operations (see Appendix), and can be digitally performed also in Alg.~\ref{alg:Kalman-analog}.
Besides, Step~2 in Alg.~\ref{alg:Kalman-digital} involves two matrix-matrix products, each with complexity $8X^3$ (see Appendix), leading to a complexity of $16X^3$.
Thus, it is replaced by step~2 in Alg.~\ref{alg:Kalman-analog}, requiring the same complexity as Alg.~\ref{alg:covariance}, i.e., $6X^2$ real operations.
In the ``correction'' phase, the computation of $\hat{\mathbf{x}}_{t\mid t}$ in Alg.~\ref{alg:Kalman-digital} through Steps~3 and 4 requires two matrix-matrix products and a matrix inversion, leading to a complexity of $8(XY^2+X^2Y+Y^3/3)$, as for the \gls{lmmse} estimator analyzed in Remark~\ref{rem:lmmse}.
To reduce complexity, these steps are replaced by Steps~3, 4, and 5 in Alg.~\ref{alg:Kalman-analog}, having a respective complexity of $8X^2$ (matrix-vector product, see Appendix), $6XY$ (same as Alg.~\ref{alg:LMMSE}, see Remark~\ref{rem:lmmse}), and $2X$ (sum of complex vectors, see Appendix).
Finally, Step~5 in Alg.~\ref{alg:Kalman-digital} requires two matrix-matrix products, and has complexity $8(X^2Y+X^3)$.
Thus, it is replaced by Steps~6 and 7 in Alg.~\ref{alg:Kalman-analog} having lower complexity, i.e., $6XY$ (same as Alg.~\ref{alg:covariance}) and $4X^2$ (same as Alg.~\ref{alg:inversion}), respectively.

\begin{remark}
A \gls{milac} can compute the Kalman filter with much lower computational complexity compared to conventional digital computing.
Specifically, the computational complexity of Alg.~\ref{alg:Kalman-analog} is given by adding the complexity of its Steps~1-7, i.e., $8X^2+6X^2+8X^2+6XY+2X+6XY+4X^2$, giving approximately $26X^2+12XY$ real operations.
In contrast, the complexity of digitally computing the Kalman filter via Alg.~\ref{alg:Kalman-digital} is given by adding the complexity of its Steps~1-5, namely $24X^3+16X^2Y+8XY^2+8Y^3/3$ real operations approximately.
The complexity of computing the Kalman filter with a \gls{milac} and by digital computing are compared in Fig.~\ref{fig:complexity-Kalman}, where we observe a gain of $1.1\times 10^4$ times when $X=Y=8192$.
This decrease in complexity could have critical implications in the field of financial forecasting, where Kalman filters are extensively used to predict the trend of markets and assets with stringent latency requirements \cite{mil24}.
\end{remark}

\begin{figure}[t]
\centering
\includegraphics[width=0.42\textwidth]{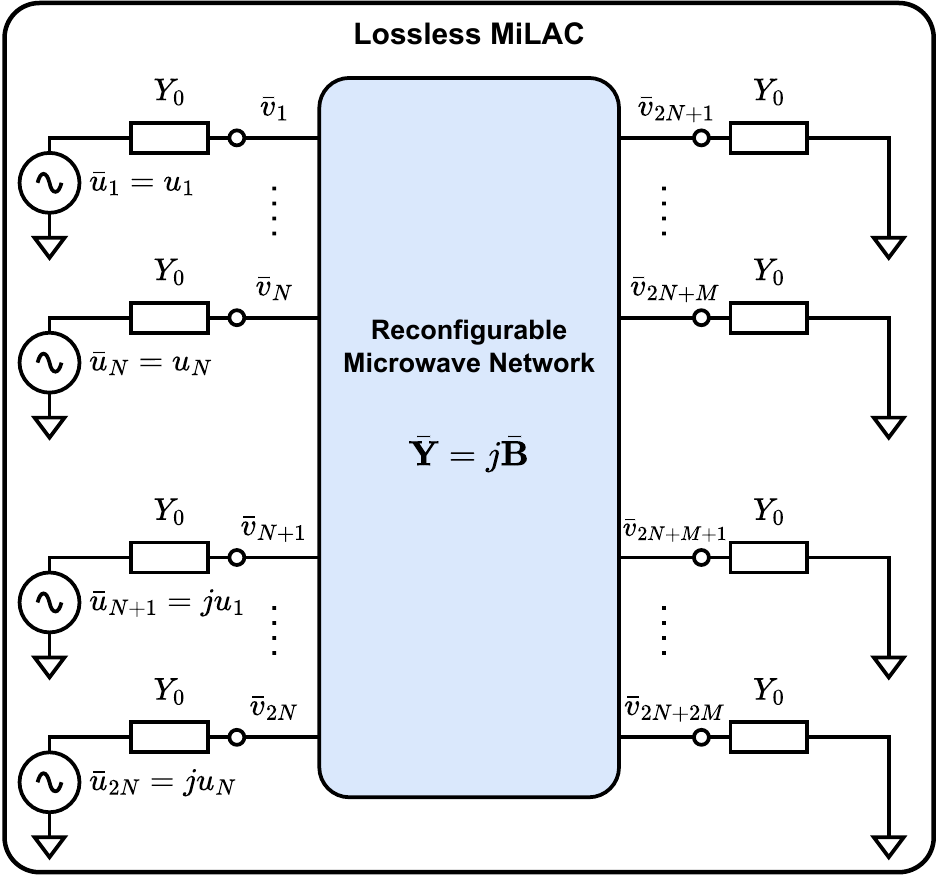}
\caption{Representation of a MiLAC implemented with a lossless microwave network.}
\label{fig:ac-imag}
\end{figure}

\section{Practical Implementation of a MiLAC with a Lossless Microwave Network}
\label{sec:implementation}

We have modeled and analyzed a \gls{milac} made of tunable admittance components allowed to take any complex value, with arbitrary real part (conductance) and imaginary part (susceptance).
However, the use of tunable admittance components with non-zero real parts in a \gls{milac} should be avoided for the following two reasons.
First, positive conductance values dissipate power, causing losses and increasing power consumption.
Second, negative conductance values require active components, which are costly and require additional power supply.
To solve these problems, in this section, we show that it is possible to realize a fully capable \gls{milac} with purely imaginary tunable admittance components, i.e., with a lossless microwave network.

Consider a \gls{milac} implemented through a lossless microwave network with an admittance matrix $\bar{\mathbf{Y}}\in\mathbb{C}^{2P\times2P}$, receiving the input on the first $2N$ ports, where $P=N+M$, $M>0$, as shown in Fig.~\ref{fig:ac-imag}.
Since it is made of purely imaginary admittance components, $\bar{\mathbf{Y}}$ is purely imaginary, and is given by $\bar{\mathbf{Y}}=j\bar{\mathbf{B}}$, where $\bar{\mathbf{B}}\in\mathbb{R}^{2P\times2P}$ is the susceptance matrix of the lossless microwave network.
We denote the voltage of the $n$th voltage source as $\bar{u}_n$ , for $n=1,\ldots,2N$, and the input vector as $\bar{\mathbf{u}}=[\bar{u}_1,\ldots,\bar{u}_{2N}]^T\in\mathbb{C}^{2N\times1}$.
Besides, we denote as $\bar{v}_p$ the voltage at the $p$th port of the microwave network, for $p=1,\ldots,2P$, and the output vector as $\bar{\mathbf{v}}=[\bar{v}_1,\ldots,\bar{v}_{2P}]^T\in\mathbb{C}^{2P\times1}$, which can be partitioned as $\bar{\mathbf{v}}=[\bar{\mathbf{v}}_1^T,\bar{\mathbf{v}}_2^T]^T$, where $\bar{\mathbf{v}}_1=[\bar{v}_1,\ldots,\bar{v}_{2N}]^T\in\mathbb{C}^{2N\times1}$ and $\bar{\mathbf{v}}_2=[\bar{v}_{2N+1},\ldots,\bar{v}_{2N+2M}]^T\in\mathbb{C}^{2M\times1}$.

According to the analysis in Section~\ref{sec:analysis1}, the output $\bar{\mathbf{v}}=[\bar{\mathbf{v}}_1^T,\bar{\mathbf{v}}_2^T]^T$ is given as a function of the input $\bar{\mathbf{u}}$ and the susceptance matrix $\bar{\mathbf{B}}$ as
\begin{equation}
\begin{bmatrix}
\bar{\mathbf{v}}_1\\
\bar{\mathbf{v}}_2
\end{bmatrix}
=
\begin{bmatrix}
\bar{\mathbf{Q}}_{11}\\
\bar{\mathbf{Q}}_{21}
\end{bmatrix}
\bar{\mathbf{u}},\label{eq:v-bar}
\end{equation}
where $\bar{\mathbf{Q}}_{11}$ and $\bar{\mathbf{Q}}_{21}$ are blocks of the matrix
\begin{equation}
\bar{\mathbf{Q}}=
\begin{bmatrix}
\bar{\mathbf{Q}}_{11} & \bar{\mathbf{Q}}_{12}\\
\bar{\mathbf{Q}}_{21} & \bar{\mathbf{Q}}_{22}
\end{bmatrix}
=\bar{\mathbf{P}}^{-1},
\end{equation}
with $\bar{\mathbf{Q}}_{11}\in\mathbb{C}^{2N\times 2N}$, $\bar{\mathbf{Q}}_{12}\in\mathbb{C}^{2N\times 2M}$, $\bar{\mathbf{Q}}_{21}\in\mathbb{C}^{2M\times 2N}$, $\bar{\mathbf{Q}}_{22}\in\mathbb{C}^{2M\times 2M}$, where
\begin{equation}
\bar{\mathbf{P}}=\frac{j\bar{\mathbf{B}}}{Y_0}+\mathbf{I}_{2P}.\label{eq:P-bar}
\end{equation}

Interestingly, it is possible to prove that the output of this lossless \gls{milac} is given by $\bar{\mathbf{v}}_1=[-j\mathbf{v}_1^T,\mathbf{v}_1^T]^T$ and $\bar{\mathbf{v}}_2=[-j\mathbf{v}_2^T,\mathbf{v}_2^T]^T$, where $\mathbf{v}_1$ and $\mathbf{v}_2$ are given by \eqref{eq:v}, when the input is set as $\bar{\mathbf{u}}=[\mathbf{u}^T,j\mathbf{u}^T]^T$ and the susceptance matrix is set as
\begin{equation}
\bar{\mathbf{B}}=
\begin{bmatrix}
\bar{\mathbf{B}}_{11} & \bar{\mathbf{B}}_{12}\\
\bar{\mathbf{B}}_{21} & \bar{\mathbf{B}}_{22}
\end{bmatrix},\label{eq:B-bar-1}
\end{equation}
where
\begin{equation}
\bar{\mathbf{B}}_{11}=
\begin{bmatrix}
 \Re\{\mathbf{Y}_{11}\} & \Im\{\mathbf{Y}_{11}\}\\
-\Im\{\mathbf{Y}_{11}\} & \Re\{\mathbf{Y}_{11}\}
\end{bmatrix}+Y_0
\begin{bmatrix}
 \mathbf{I}_{N} & \mathbf{I}_{N}\\
-\mathbf{I}_{N} & \mathbf{I}_{N}
\end{bmatrix},\label{eq:B-bar-2}
\end{equation}
\begin{equation}
\bar{\mathbf{B}}_{22}=
\begin{bmatrix}
 \Re\{\mathbf{Y}_{22}\} & \Im\{\mathbf{Y}_{22}\}\\
-\Im\{\mathbf{Y}_{22}\} & \Re\{\mathbf{Y}_{22}\}
\end{bmatrix}+Y_0
\begin{bmatrix}
 \mathbf{I}_{M} & \mathbf{I}_{M}\\
-\mathbf{I}_{M} & \mathbf{I}_{M}
\end{bmatrix},\label{eq:B-bar-3}
\end{equation}
\begin{equation}
\bar{\mathbf{B}}_{ik}=
\begin{bmatrix}
 \Re\{\mathbf{Y}_{ik}\} & \Im\{\mathbf{Y}_{ik}\}\\
-\Im\{\mathbf{Y}_{ik}\} & \Re\{\mathbf{Y}_{ik}\}\\
\end{bmatrix},
\label{eq:B-bar-4}
\end{equation}
for $ik\in\{12,21\}$, as represented in Fig.~\ref{fig:ac-imag}.
Thus, it is possible to implement the \gls{milac} analyzed in Section~\ref{sec:analysis1} and compute the functions in \eqref{eq:v1-1}-\eqref{eq:v2-1} and \eqref{eq:v1-2}-\eqref{eq:v2-2} by only employing lossless tunable admittance components.
To prove this result, we need to show that it holds
\begin{equation}
\left(\frac{j\bar{\mathbf{B}}}{Y_0}+\mathbf{I}_{2P}\right)
\begin{bmatrix}
-j\mathbf{v}_1\\
\mathbf{v}_1\\
-j\mathbf{v}_2\\
\mathbf{v}_2
\end{bmatrix}=
\begin{bmatrix}
\mathbf{u}\\
j\mathbf{u}\\
\mathbf{0}\\
\mathbf{0}
\end{bmatrix},\label{eq:lossless-proof-1}
\end{equation}
when $\bar{\mathbf{B}}$ is set as \eqref{eq:B-bar-1}-\eqref{eq:B-bar-4}.
Instead of proving \eqref{eq:lossless-proof-1}, we can introduce the permutation matrix $\boldsymbol{\Pi}\in\mathbb{C}^{2P\times2P}$ as
\begin{equation}
\boldsymbol{\Pi}=
\begin{bmatrix}
\mathbf{I}_{N} & \mathbf{0}_{N} & \mathbf{0}_{N\times M} & \mathbf{0}_{N\times M}\\
\mathbf{0}_{M\times N} & \mathbf{0}_{M\times N} & \mathbf{I}_{M} & \mathbf{0}_{M}\\
\mathbf{0}_{N} & \mathbf{I}_{N} & \mathbf{0}_{N\times M} & \mathbf{0}_{N\times M}\\
\mathbf{0}_{M\times N} & \mathbf{0}_{M\times N} & \mathbf{0}_{M} & \mathbf{I}_{M}
\end{bmatrix},
\end{equation}
and show that it holds
\begin{equation}
\boldsymbol{\Pi}\left(\frac{j\bar{\mathbf{B}}}{Y_0}+\mathbf{I}_{2P}\right)\boldsymbol{\Pi}^T\boldsymbol{\Pi}
\begin{bmatrix}
-j\mathbf{v}_1\\
\mathbf{v}_1\\
-j\mathbf{v}_2\\
\mathbf{v}_2
\end{bmatrix}
=\boldsymbol{\Pi}
\begin{bmatrix}
\mathbf{u}\\
j\mathbf{u}\\
\mathbf{0}\\
\mathbf{0}
\end{bmatrix},\label{eq:lossless-proof-2}
\end{equation}
which is equivalent to \eqref{eq:lossless-proof-1} since $\boldsymbol{\Pi}^T\boldsymbol{\Pi}=\mathbf{I}_{2P}$.
By recalling $\mathbf{v}=[\mathbf{v}_1^T,\mathbf{v}_2^T]^T$ and $\tilde{\mathbf{u}}=[\mathbf{u}^T,\mathbf{0}_{M\times1}^T]^T$, \eqref{eq:lossless-proof-2} becomes
\begin{equation}
\boldsymbol{\Pi}\left(\frac{j\bar{\mathbf{B}}}{Y_0}+\mathbf{I}_{2P}\right)\boldsymbol{\Pi}^T
\begin{bmatrix}
-j\mathbf{v}\\
\mathbf{v}\\
\end{bmatrix}=
\begin{bmatrix}
\tilde{\mathbf{u}}\\
j\tilde{\mathbf{u}}\\
\end{bmatrix},\label{eq:lossless-proof-3}
\end{equation}
in which the term $\boldsymbol{\Pi}(j\bar{\mathbf{B}}/Y_0+\mathbf{I}_{2P})\boldsymbol{\Pi}^T$ can be rewritten as
\begin{align}
\boldsymbol{\Pi}\left(\frac{j\bar{\mathbf{B}}}{Y_0}+\mathbf{I}_{2P}\right)\boldsymbol{\Pi}^T
=&\frac{j}{Y_0}\boldsymbol{\Pi}\bar{\mathbf{B}}\boldsymbol{\Pi}^T+\mathbf{I}_{2P}\label{eq:lossless-term-1}\\
=&\frac{j}{Y_0}
\begin{bmatrix}
 \Re\{\mathbf{Y}\} & \Im\{\mathbf{Y}\}\\
-\Im\{\mathbf{Y}\} & \Re\{\mathbf{Y}\}
\end{bmatrix}\label{eq:lossless-term-2}\\
&+j
\begin{bmatrix}
 \mathbf{I}_{P} & \mathbf{I}_{P}\\
-\mathbf{I}_{P} & \mathbf{I}_{P}
\end{bmatrix}
+\mathbf{I}_{2P},\label{eq:lossless-term-3}
\end{align}
where in \eqref{eq:lossless-term-1} we used $\boldsymbol{\Pi}\boldsymbol{\Pi}^T=\mathbf{I}_{2P}$ and in \eqref{eq:lossless-term-2}-\eqref{eq:lossless-term-3} we rewrote $\boldsymbol{\Pi}\bar{\mathbf{B}}\boldsymbol{\Pi}^T$ recalling that $\bar{\mathbf{B}}$ is set as in \eqref{eq:B-bar-1}-\eqref{eq:B-bar-4}.
Substituting \eqref{eq:lossless-term-1}-\eqref{eq:lossless-term-3} into \eqref{eq:lossless-proof-3} and executing the products, \eqref{eq:lossless-proof-3} simplifies as
\begin{equation}
\frac{1}{Y_0}
\begin{bmatrix}
\mathbf{Y}\mathbf{v}\\
j\mathbf{Y}\mathbf{v}
\end{bmatrix}
+
\begin{bmatrix}
\mathbf{v}\\
j\mathbf{v}
\end{bmatrix}
=
\begin{bmatrix}
\tilde{\mathbf{u}}\\
j\tilde{\mathbf{u}}
\end{bmatrix},
\end{equation}
which holds true given the relationship in \eqref{eq:v-tmp}-\eqref{eq:P}, proving that a fully capable \gls{milac} as analyzed in Section~\ref{sec:analysis1} can be implemented with lossless microwave network.
Notably, a drawback of using lossless admittance components, which have purely imaginary values, is that the number of required components increases to $(2P)^2$, compared to the $P^2$ components needed to implement a \gls{milac} with complex admittance components.

A similar discussion is valid for a \gls{milac} with input applied on all the ports, i.e., with $M=0$.
In detail, considering a \gls{milac} implemented with a lossless network as in Fig.~\ref{fig:ac-imag} but with $M=0$, the analysis in Section~\ref{sec:analysis2} gives that
\begin{equation}
\bar{\mathbf{v}}_1=\bar{\mathbf{P}}^{-1}\bar{\mathbf{u}},
\end{equation}
where
\begin{equation}
\bar{\mathbf{P}}=\frac{j\bar{\mathbf{B}}}{Y_0}+\mathbf{I}_{2N}.
\end{equation}
Thus, it is possible to show that we obtain $\bar{\mathbf{v}}_1=[-j\mathbf{v}_1^T,\mathbf{v}_1^T]^T$, where $\mathbf{v}_1$ is given by \eqref{eq:v-M0}-\eqref{eq:P-M0}, when the input is set as $\bar{\mathbf{u}}=[\mathbf{u}^T,j\mathbf{u}^T]^T$ and the susceptance matrix is set as
\begin{equation}
\bar{\mathbf{B}}=
\begin{bmatrix}
\Re\{\mathbf{Y}\} & \Im\{\mathbf{Y}\}\\
-\Im\{\mathbf{Y}\} &  \Re\{\mathbf{Y}\}
\end{bmatrix}+Y_0
\begin{bmatrix}
 \mathbf{I}_{N} & \mathbf{I}_{N}\\
-\mathbf{I}_{N} & \mathbf{I}_{N}
\end{bmatrix}.
\end{equation}
The proof of this result is omitted for conciseness as it is similar to the proof for the case $M>0$.
Consequently, we can implement the \gls{milac} analyzed in Section~\ref{sec:analysis2} to compute \eqref{eq:v-M0}-\eqref{eq:P-M0} by only using lossless tunable admittance components.

\section{Conclusion}
\label{sec:conclusion}

Analog computing is undergoing a revival, promising low-power and massively parallel computation capabilities for applications in signal processing and communications.
In Part~I of this paper, we introduce the concept of \gls{milac} as a generic analog computer that linearly operates with microwave signals, and analyze its application to signal processing.
A \gls{milac} can generally be modeled as a multiport microwave network made of tunable impedance (or admittance) components.
At the ports of such a microwave network, the input signals are applied and the output signals are read.
By using rigorous microwave theory, we derive the expression of the output of a \gls{milac} as a function of its input and the values of its admittance components.
We show that a \gls{milac} can efficiently compute the \gls{lmmse} estimator and invert an arbitrary matrix in the analog domain, with significantly reduced computational complexity over conventional digital computing.
In particular, matrix inversion can be performed with complexity growing with the square of the matrix size, rather than the cube as with digital computing.
The primitive operations computable with a \gls{milac} can serve as building blocks for implementing more complex algorithms, such as the Kalman filter, directly in the analog domain, achieving a substantial decrease in computational complexity.
While we have shown how to efficiently compute certain common operations with \gls{milac}, future research could investigate whether other operations can also be efficiently computed in the analog domain with reconfigurable microwave networks.
Finally, to facilitate the practical implementation of a \gls{milac}, we present a design based on lossless admittance components, avoiding losses and costly active \gls{rf} components.
Further future research could investigate the computational capabilities of microwave networks in the presence of hardware non-idealities, such as lossy interconnections.
For instance, this can be done departing from previous literature on \gls{bd-ris} \cite{del25}, where fully-connected microwave networks have been represented by separating the effects of the interconnections and the tunable components.
Also assessing the power consumption of \gls{milac} and validating its practical benefits are interesting future research directions.

In Part~II of this paper, we discuss the application of \gls{milac} to wireless communications, showing that it can enable gigantic \gls{mimo} communications by performing computations and beamforming in the analog domain.

\section*{Appendix}

In this appendix, we review the computational complexity of relevant matrix operations performed with classical digital computing.
We define the computational complexity of an operation as the number of arithmetic real operations required to complete it, where these operations include addition, subtraction, multiplication, and division.
Following this definition, we assume that computing the transpose or the conjugate transpose of a matrix requires no operations.
For operations on complex numbers, we recall that addition and subtraction require two real operations, while multiplication requires six real operations.

\subsubsection{Scalar Product}

Given two vectors $\mathbf{a}\in\mathbb{C}^{1\times N}$ and $\mathbf{b}\in\mathbb{C}^{N\times 1}$, the scalar product $\mathbf{a}\mathbf{b}$ requires $N-1$ complex additions and $N$ complex multiplications, since it is given by $[\mathbf{a}]_1[\mathbf{b}]_1+\ldots+[\mathbf{a}]_N[\mathbf{b}]_N$.
Considering that complex additions and multiplications require two and six real operations, respectively, the complexity of the scalar product is approximately $8N$ real operations.

\subsubsection{Matrix-Vector Product}

Given a matrix $\mathbf{A}\in\mathbb{C}^{M\times N}$ and a vector $\mathbf{b}\in\mathbb{C}^{N\times 1}$, the matrix-vector product $\mathbf{A}\mathbf{b}$ can be expressed as $\mathbf{A}\mathbf{b}=[\mathbf{a}_1\mathbf{b},\ldots,\mathbf{a}_M\mathbf{b}]^T$, where $\mathbf{a}_m\in\mathbb{C}^{1\times N}$ is the $m$th row of $\mathbf{A}$, for $m=1,\ldots,M$.
Thus, it requires $M$ times the complexity of the scalar product $\mathbf{a}\mathbf{b}$, which is approximately $8MN$ real operations.

\subsubsection{Matrix-Matrix Product}

Given two matrices $\mathbf{A}\in\mathbb{C}^{M\times N}$ and $\mathbf{B}\in\mathbb{C}^{N\times L}$, the matrix-matrix product $\mathbf{A}\mathbf{B}$ can be expressed as $\mathbf{A}\mathbf{B}=[\mathbf{A}\mathbf{b}_1,\ldots,\mathbf{A}\mathbf{b}_L]$, where $\mathbf{b}_\ell\in\mathbb{C}^{N\times 1}$ is the $\ell$th column of $\mathbf{B}$, for $\ell=1,\ldots,L$.
Thus, it requires $L$ times the complexity of the matrix-vector product $\mathbf{A}\mathbf{b}$, which is approximately $8LMN$ real operations.

\subsubsection{Matrix Inversion}

Given an invertible matrix $\mathbf{A}\in\mathbb{C}^{N\times N}$, the algorithm commonly used to compute the matrix inverse $\mathbf{A}^{-1}$ is the Gauss's method \cite[Chapter 1]{far88}.
Thus, we regard the complexity of matrix inversion as equal to the complexity of the Gauss's method algorithm.
It can be shown following \cite[Chapter 1]{far88} that the Gauss’s method requires $N(N+1)/2$ divisions, $(2N^3+3N^2-5N)/6$ multiplications, and $(2N^3+3N^2-5N)/6$ subtractions.
Consequently, the algorithm’s complexity is dominated by approximately $N^3/3$ multiplications and $N^3/3$ subtractions, which require a total of $8N^3/3$ real operations.

\bibliographystyle{IEEEtran}
\bibliography{IEEEabrv,main}

\end{document}

%% file: glossary.tex
\newacronym{iid}{i.i.d.}{independent and identically distributed}
\newacronym{mimo}{MIMO}{multiple-input multiple-output}
\newacronym{miso}{MISO}{multiple-input single-output}
\newacronym{adc}{ADC}{analog-to-digital converter}
\newacronym{dac}{DAC}{digital-to-analog converter}
\newacronym{rf}{RF}{radio-frequency}
\newacronym{mmwave}{mmWave}{millimeter wave}
\newacronym{dft}{DFT}{discrete Fourier transform}
\newacronym{doa}{DoA}{direction-of-arrival}
\newacronym{awgn}{AWGN}{additive white Gaussian noise}
\newacronym{snr}{SNR}{signal-to-noise ratio}
\newacronym{sqnr}{SQNR}{signal-to-quantization-noise ratio}
\newacronym{qpsk}{QPSK}{quadrature phase shift keying}
\newacronym{ber}{BER}{bit error rate}
\newacronym{csi}{CSI}{channel state information}
\newacronym{em}{EM}{electromagnetic}

\newacronym{espar}{ESPAR}{electronically steerable parasitic array radiator}
\newacronym{lma}{LMA}{load modulated array}

\newacronym{mse}{MSE}{mean square error}
\newacronym{mmse}{MMSE}{minimum mean square error}
\newacronym{lmmse}{LMMSE}{linear minimum mean square error}
\newacronym{gls}{GLS}{generalized least squares}
\newacronym{gmf}{GMF}{generalized matched filtering}
\newacronym{rls}{RLS}{regularized least squares}
\newacronym{ols}{OLS}{ordinary least squares}
\newacronym{omf}{OMF}{ordinary matched filtering}

\newacronym{zf}{ZF}{zero-forcing}
\newacronym{ls}{LS}{least squares}
\newacronym{mf}{MF}{matched filtering}
\newacronym{r-zfbf}{R-ZFBF}{regularized zero-forcing beamforming}
\newacronym{zfbf}{ZFBF}{zero-forcing beamforming}
\newacronym{mbf}{MBF}{matched beamforming}

\newacronym{ris}{RIS}{reconfigurable intelligent surface}
\newacronym{bd-ris}{BD-RIS}{beyond-diagonal RIS}
\newacronym{sim}{SIM}{stacked intelligent metasurface}
\newacronym{dsp}{DSP}{digital signal processor}
\newacronym{asp}{ASP}{analog signal processing}

\newacronym{milac}{MiLAC}{microwave linear analog computer}

\newacronym{isa}{ISA}{instruction set architecture}

\newacronym{nn}{NN}{neural network}

\newacronym{cpu}{CPU}{central processing unit}

\newacronym{dma}{DMA}{dynamic metasurface antenna}
\newacronym{pwm}{PWM}{pulse-width modulation}